\documentclass[submission,copyright,creativecommons]{eptcs}
\providecommand{\event}{QAPL 2019} % Name of the event you are submitting to
\usepackage{breakurl}             % Not needed if you use pdflatex only.
\usepackage{underscore}           % Only needed if you use pdflatex.

\usepackage{amsthm}
\newcommand\mydots{\makebox[1em][c]{.\hfil.\hfil.}}
\newtheorem{theorem}{Theorem}%[theorem]
\newtheorem{corollary}[theorem]{Corollary}
\newtheorem{lemma}[theorem]{Lemma}
\newtheorem{example}[theorem]{Example}
\newtheorem{definition}[theorem]{Definition}
\newtheorem{proposition}[theorem]{Proposition}

\usepackage{todonotes}
\usepackage{fancyvrb}
\newcommand{\prg}{|\mathtt{prg}|}
 
\newcommand{\pre}{\mathit{pre}}
\newcommand{\img}{\mathit{img}}

\newcommand{\stcomp}[1]{\widetilde{#1}}
\newcommand{\stcompc}[1]{{#1}^{\raisebox{-1pt}{$\scriptstyle\complement$}}}

\newcommand{\abs}[1]{\uparrow\!\!#1}

\def\supsharp{^{\raisebox{-5pt}{$\scriptstyle\sharp$}}}
\def\supflat{^{\raisebox{-5pt}{$\scriptstyle\flat$}}}
\usepackage[all]{foreign}

\usepackage{paralist}

\usepackage{galois}

\usepackage{amssymb}
\usepackage{latexsym}

\usepackage{stmaryrd}
\usepackage{mathtools}

\usepackage{thmtools} 
\usepackage{thm-restate}
\usepackage{placeins} % \floatBarrier

\usepackage{tikz}
\usetikzlibrary{matrix,positioning,calc}
\usetikzlibrary{positioning,shapes,fit,arrows}
\usetikzlibrary{shapes.geometric,positioning}
\usepackage{times}
\usetikzlibrary{mindmap,backgrounds}
\usepackage{rotating}
\usepackage{pgfplots}
    \pgfplotsset{compat=1.3,
        layers/my layer set/.define layer set={
            background,
            main,
            foreground,
            forforeground
        }{
        },
        set layers=my layer set,
    }

\definecolor{mycolor}{rgb}{0.02,0.4,0.7}
\definecolor{mycolor2}{rgb}{0.0,0.5,0.0}
\definecolor{mycolor3}{rgb}{0.7,0.4,0.02}
\definecolor{mycolor4}{rgb}{0.7,0.02,0.4}
\usepackage{colortbl}

\usepackage{caption}
\usepackage{subcaption}
\PassOptionsToPackage{pdftex}{graphicx}
\usepackage{graphicx}

\title{Probabilistic Output Analyses for Deterministic Programs\\ --Reusing Existing Non-probabilistic Analyses}
\author{Maja Hanne Kirkeby\thanks{This work is supported by The Danish Council for Independent Research, Natural Sciences, grant no.~DFF 4181-00442.}
\institute{Roskilde University, Denmark}
\email{kirkebym@acm.org}
}
\def\titlerunning{Probabilistic Output Analyses}
\def\authorrunning{Maja H.\ Kirkeby}

\begin{document}
\maketitle

\begin{abstract}
%Given a deterministic program, the probabilities of the program inputs, and a set of outputs, \ie, an output event, we will analyse the probability of the output event. 
%For deterministic programs, the probability of an output event depends on (i) the inputs that lead to the event, \ie, its pre-image; and (ii) the probability of those inputs. 
We consider reusing established non-probabilistic output analyses (either forward or backwards) that yield over-approximations of a program's pre-image or image relation, e.g., interval analyses. We assume a probability measure over the program input and present two techniques (one for forward and one for backward analyses) that both derive upper and lower probability bounds for the output events. We demonstrate the most involved technique, namely the forward technique, for two examples and compare their results to a cutting-edge probabilistic output analysis.
\end{abstract}

\section{Introduction}
Output analyses infer information about program outputs either as main purpose, e.g., interval analysis~\cite{Cousot1977} and octagon analysis~\cite{Mine2006}, or as bi-product, \eg, sign-analysis~\cite{Nielson2007}. Output analyses are also used to construct other analyses, e.g., resource analyses where resource instrumented versions of the program are analysed using output analyses~\cite{Kirkeby2015}.
\emph{The aim of this paper} is to reuse static (non-probabilistic) output analyses to infer information about the probabilities of a program's output when knowing the input probabilities, i.e., providing an approach to mechanically obtain probabilistic output analyses for deterministic programs.  
% We want to answer questions such as \emph{``what is the probability that a given program yields a result between 0 and 10?''} and \emph{``what is the probability that a given program uses more than 10 resource units, \eg{}, time steps, before it stops?''}. 

Previously, probabilistic analyses have mainly focused on analysing probabilistic programs. Sankaranarayan \etal{}~\cite{conf:pldi:SankaranarayananC2013} and Adje \etal{}~\cite{Adje2013} present analyses optimized for each their type(s) of probability measures and both provide upper and lower probability bounds of the program output. In this paper, we focus on a subset of probabilistic programs, namely the deterministic ones, but instead of presenting a specified probability analysis, we present an approach to reuse analyses to create probabilistic analyses for any type of probability measures.

More general results exist in the form of probabilistic abstract interpretation frameworks presented by Monniaux~\cite{Monniaux2000} and by Cousot and Monerau~\cite{CousotM12}. 
Both these frameworks describe how to extend non-probabilistic abstract interpretation analyses for deterministic programs to probabilistic analyses for probabilistic programs; their resulting analyses provide upper probability bounds of output events (and not lower probability bounds). Both require a manual development to handle the randomness in the programs, e.g., random number generators and the probabilistic operations. 
In comparison, the techniques presented in this paper handle only deterministic programs (disallowing, e.g., random number generators) but this choice allows us to consider the existing analysis tools as black-box analyses and it does not require any manual developments. Thus, the techniques are amenable to implementation.
Furthermore, they each induce not only upper probability bounds of output events, but also lower probability bounds. We will compare our results with those produced by Monniaux's experimental lifting of an interval analysis~\cite{Monniaux2000}.
\paragraph{Contributions and Overview. }%\todo{opdater til sidst}
After preliminaries (Section \ref{sec:preliminaries}), we present two novel techniques for inducing both upper and lower probability bounds of output events of deterministic programs: one using backwards analyses (Section \ref{sec:backward}) and one using forward analyses (Section \ref{sec:forward}). We demonstrate the forward technique, the most involved of the two, by two examples (Section \ref{sec:casestudyinterval}); one using sign analysis in combination with termination analysis, and one using interval analysis. When comparing the probability bounds we produce with the ones 
 produced by Monniaux's experimental probabilistic analysis~\cite{Monniaux2000}, the presented approaches infer equally good and better upper probability bounds. Furthermore, when combined with non-termination analyses, they produce novel non-trivial lower bounds, \ie, bounds greater than 0.

\section{Preliminaries}\label{sec:preliminaries}
%\chapter{Probabilities and other prerequisites}
%\label{chp:introProb}

%%%%%%%%%%%%%
We let $X, Y, A, B, T$ (sometimes indexed) denote sets. A set $\stcompc{A}$ is the complement of $A$ with respect to a set (indicated by the context), e.g., $X$. %A set $X$ is \emph{base} of its power set $\wp(X)$. 
We denote a \emph{countable infinite series} by $X_1, X_2, \ldots$. 
A \emph{partition} $T$ over a set $X$ is a set of nonempty and pairwise disjoint subsets of $X$ such that $\bigcup T = X$; when $T$ is  a finite/countable/infinite set then $T$ is a \emph{finite/countable/infinite partition}. For two partitions $T,T'$ over $X$, we say that $T$ is \emph{finer} than $T'$ if every element of $T$ is a subset of an element in $T'$. When $T$ over $X$ consists of all singletons, \ie, $T=\{\{x\}\mid x\in X\}$, it is a \emph{singleton-partition}; the finest partition $T$ over $X$ is the singleton-partition.

%A \emph{probability distribution} is defined by a function $P\colon X \rightarrow [0,1]$ over a countable set $X$ for which $\sum_{x\in X}P(x) = 1$.
A \emph{$\sigma$-algebra} $\mathcal{X}$ is a non-empty subset of $\wp(X)$ that is closed under countable unions and complements, \ie{} if $A_1, A_2, \ldots \in \mathcal{X}$, then $\bigcup^{\infty}_{n=0} A_n \in \mathcal{X}$, and if $A\in \mathcal{X}$, then $\stcompc{A} \in\mathcal{X}$. Note that if $\mathcal{X}$ is a $\sigma$-algebra over $X$, then $X \in  \mathcal{X}$ and $\emptyset \in  \mathcal{X}$, and, furthermore, that a $\sigma$-algebra $\mathcal{X}$ is closed under countable intersections, \ie{}  if $A_1, A_2, \ldots \in \mathcal{X}$, then $\bigcap^{\infty}_{n=0} A_n \in \mathcal{X}$. 
%
%When $X$ is a set, and $\mathbf{C}$ is a collection of $\sigma$-algebras over $X$; then, $\bigcap \mathbf{C}$ is a $\sigma$-algebra over $X$. Furthermore, 
Given a collection of sets $A \subseteq \wp(X)$, the $\sigma$-algebra \emph{generated}  by $A$, written $\sigma(A)$, is the intersection of all $\sigma$-algebras containing $A$.

A \emph{measurable space} is a pair $(X,\mathcal{X})$ whereby \emph{the sample space} $X$ is a set and $\mathcal{X} \subseteq \wp(X)$ is a $\sigma$-algebra.
The elements of $\mathcal{X}$ are called \emph{events}. % or sometimes \emph{measurable subsets of $X$}.
A \emph{measure} $\mu$ on a measurable space $(X,\mathcal{X})$ is a function $\mu\colon \mathcal{X} \rightarrow \mathbb{R}^+$ that is countably additive, \ie{} for every countable set of pairwise disjoint sets $A_1, A_2, \ldots \in \mathcal{X}$, $\mu(\cup^{\infty}_{i=1}A_i) = \sum^{\infty}_{i=1} \mu(A_i)$, and $\mu(\emptyset)=0$. 
Note that a {measure} $\mu\colon \mathcal{X} \rightarrow \mathbb{R}^+$ is monotone, \ie{}, $A \subseteq B \Rightarrow \mu(A) \leq \mu(B)$, whenever $A, B \in \mathcal{X}$.
A \emph{measure space} $(X,\mathcal{X},\mu)$ is a measurable space $(X,\mathcal{X})$ with a measure $\mu$ on it.
Given two measure spaces $(X,\mathcal{X},\mu_{\mathcal{X}})$ and $(Y,\mathcal{Y},\mu_{\mathcal{Y}})$, their \emph{product measure space} is $(X \times Y, \sigma( \mathcal{X} \times \mathcal{Y}), \mu)$ where
$\mu(A\times B) =  \mu_{\mathcal{X}}(A) \cdot \mu_{\mathcal{Y}}(B)$ for all $A \in \mathcal{X}$ and $B \in \mathcal{Y}$.
A measure $\mu$ on a measurable space $(X,\mathcal{X})$ is \emph{discrete} if its weight is on at most countably many elements, \ie{} there exists a countable set $A \in \mathcal{X}$ such that $\mu(A)=\mu(X)$, and  \emph{continuous} if the weights of all countable sets are 0, \ie{} $\mu(A)=0$ for all countable sets $A \in \mathcal{X}$.

A measure $\mu$ on $(X,\mathcal{X})$ is a \emph{probability measure}
if $\mu(X) = 1$;
in addition, $\mu(A) = 1 -\mu(\stcompc{A})$.
A \emph{probability space} $(X,\mathcal{X},\mu)$ is a measure space wherein the measure $\mu$ is a probability measure.
%If $(X,\mathcal{X})$ is a measurable space; then, a \emph{probability distribution} $P\colon X \rightarrow [0,1]$ defines a (discrete) probability space $(X,\mathcal{X},\mu)$, where $\mu(A) = \sum_{x\in A}P(x)$ whenever $A \in \mathcal{X}$.

Let $(X,\mathcal{X})$ and $(Y,\mathcal{Y})$ be measurable spaces:
A function $f\colon X \rightarrow Y$ is \emph{measurable} if $\mathit{pre}_{f}(B) \in \mathcal{X}$ whenever $B \in \mathcal{Y}$, where the \emph{pre-image function of $f$} $\mathit{pre}_{f}$ is defined by $\mathit{pre}_{f}(B) \triangleq \{x \in X \mid f(x) \in B\}$ and the \emph{image function} $\img_f(A) \triangleq \{f(x) \in Y \mid x \in A\}$.  
Often we denote a measurable function by $f\colon (X,\mathcal{X}) \rightarrow (Y,\mathcal{Y})$ and we refer to elements in $\mathcal{X}$ as \emph{input events} and elements of $\mathcal{Y}$ as \emph{output events}
A probability space $(X,\mathcal{X},\mu)$ and a measurable function $f\colon (X,\mathcal{X}) \rightarrow (Y,\mathcal{Y})$, defines a probability measure $\mu_f$, called an \emph{output probability measure}, $\mu_f(A) \triangleq \mu(\mathit{pre}_{f}(A))$ of $f$ whenever $A \in \mathcal{Y}$, {\eg, \cite{Butler2018}}. 

\section{Reusing existing analysis}\label{sec:reusingexistinganalysis}
A program \texttt{prg} may have many semantics, but in this work, we consider the semantics $\prg$ to be a relation between input $X$ and output $Y$, \ie{}, a set of input-output pairs $\prg \subseteq X \times Y$.    
For a deterministic program, the input-output relation is \emph{functional}, \ie{}, each input is related to at most one output.  
Programs that terminate for all inputs $x\in X$ define \emph{total} relations,
 \ie each input relates to at least one output.  
Depending on the analysed language, $Y$ could contain special elements for program results that are not per se outputs, for instance, error or nontermination. Without loss of generality, we limit the focus of this paper to the class of programs with total input-output relations. 
Thus, assuming the program semantics to be both measurable ~\cite{Kozen1985,Monniaux2000}, \ie $\prg \colon (X,\mathcal{X}) \rightarrow (Y,\mathcal{Y})$, and total. %Furthermore, we assume an input probability measure $\mu\colon \mathcal{X} \rightarrow [0,1]$
%Let a program \texttt{prg} from input $X$ to output $Y$ have a 
%semantics $\prg\colon  X \rightarrow Y$ that is a total and measurable function, \ie $\prg \colon (X,\mathcal{X}) \rightarrow (Y,\mathcal{Y})$. 
\begin{example}\label{ex:f:total+measurable}
Let $(X,\{\emptyset,X\})$ and $(Y,\wp(Y))$ be measurable spaces where $X=\{a,b\}$ and $Y=\{c,d\}$ and let
% $f\colon (\{a,b\},\{\emptyset,\{a,b\}\}) \rightarrow (\{c,d\},\wp(\{b,c\})$ be
$f\!\colon\! X \! \rightarrow\! Y$ be a function whereby $f(a) {=} f(b) {=} c$. 
 The function $f$ is total since it is defined for each input, \ie, $a$ and $b$, and it is measurable since the pre-images of every output event, \ie, $A \in \{\emptyset,\{c\},\{d\},\{c,d\}\}$, is an input event, \ie, $\pre_f(A)\in \{\emptyset,\{a,b\}\}$;  $\pre_{f}(\emptyset)=\pre_{f}(\{d\})=\emptyset$ and $\pre_{f}(\{c\})=\pre_{f}(\{c,d\})=\{a,b\}$.
 \end{example}

According to our aim, we assume to know the input probability measure $\mu\colon \{\emptyset,X\} \rightarrow [0,1]$. 
Based on such an input probability measure $\mu$, we recall that the probability of an output event $A \in \mathcal{Y}$ is defined as the input probability of $A$'s pre-image $\pre_{\prg}(A)$, namely, $\mu(\pre_{\prg}(A))$. 
\begin{example}[Example \ref{ex:f:total+measurable} continued]\label{ex:f:outputmeasure} Let input and output spaces and the total measurable function $f$ be as in Example~\ref{ex:f:total+measurable}.
In addition, let $\mu\colon \mathcal{X} \rightarrow [0,1]$ be a trivial input probability measure such that $\mu(\emptyset)=0$ and $\mu(\{a,b\})=0$.     
 The probability of the output events $\{\emptyset,\{c\},\{d\},\{c,d\}\}$ are $\mu_f(\emptyset) = \mu_f(\{d\})= 0$ and $\mu_f(\{c\}) = \mu_f(\{c,d\}) = 1$, e.g., $\mu_f(\{d\}) = \mu(\pre_f(\{d\})) = 
\mu(\emptyset)= 0$ or $\mu_f(\{c\}) = \mu(\pre_f(\{c\})) = 
\mu(\{a,b\})= 1$.
\end{example}
 
%
%This paper is based on the idea of ``reusing an existing analysis'' and analysis $\prg\supsharp$ is typically given, \eg, using abstract interpretation~\cite{Cousot1977,Cousot1979}, in some abstract domain using an abstraction $\alpha$  from the concrete domain to the abstract domain, but we assume a concretization function $\gamma$  to avoid complications of the abstract domain. For instance, a forward interval analysis is actually a function between intervals $\prg\supsharp\colon \mathcal{I} \rightarrow \mathcal{I}$ rather than $\img_{\prg}\supsharp\colon \wp(\mathcal{R}) \rightarrow \wp(\mathcal{R})$; however, we assume to compose it with the concretization $\gamma\colon \mathcal{I}\rightarrow \wp(\mathcal{R})$ and the abstraction $\alpha\colon \wp(\mathcal{R}) \rightarrow \mathcal{I}$, \ie, $\img_{\prg}\supsharp = \gamma \comp \prg\supsharp \comp \alpha$.
This paper is based on the idea of ``reusing an existing analysis'' and an analysis $f$ is typically given in some abstract domain, \eg, using abstract interpretation~\cite{Cousot1977,Cousot1979}, using an abstraction $\alpha$  from the concrete domain to the abstract domain. We will, in addition, assume a concretization function $\gamma$  to avoid complications of the abstract domain. For instance, a forward interval analysis is actually a function between intervals $f\colon \mathcal{I} \rightarrow \mathcal{I}$ rather than a function between sets of reals $g\colon \wp(\mathbb{R}) \rightarrow \wp(\mathbb{R})$; however, we assume to compose $f$ with the concretization $\gamma\colon \mathcal{I}\rightarrow \wp(\mathbb{R})$ and the abstraction $\alpha\colon \wp(\mathbb{R}) \rightarrow \mathcal{I}$, achieving the analysis, \ie, $g = \gamma \comp f \comp \alpha$.
The analyses may be either forward or backwards; %, to obtain the over-approximation $\pre_{\prg}\supsharp$ of the pre-image $\pre_{\prg}$. 
we consider the analyses to be given as perhaps non-measurable functions $\pre_{\prg}\supsharp\colon \wp(Y) \rightarrow \wp(X)$ (backwards) or functions $\img_{\prg}\supsharp\colon \wp(X) \rightarrow \wp(Y)$ (forwards) that produce supersets\footnote{A set $A$ is a superset of a set $B$ if $A \supseteq B$.} of the programs pre-image $\pre_{\prg}$ and image $\img_{\prg}$, respectively. 
\subsection{Backwards analysis}\label{sec:backward}
In this section, we assume a pre-image over-approxi\-mating backwards analysis of the program, \eg ~\cite{cousot2013}, \ie, we assume a function  $\pre\supsharp_{\prg}\colon \wp(Y) \rightarrow \wp(X)$ such that $\pre_{\prg}(A) \subseteq \pre\supsharp_{\prg}(A)$. We want to use $\pre_{\prg}\supsharp$ to provide upper and lower probability bounds for all output events. 
%  
% We start with an order between the pre-image functions.
We start by defining an order between the pre-image functions based on the relationship of their outputs.
% We will need the following order between the pre-image functions. % based on the relationship of their outputs
\begin{definition}%[over and under-approximation of a function]
Let $\pre,\pre'\colon \wp(Y) \rightarrow \wp(X)$ be functions. 
We say that function $\pre'$ \emph{over-approximates} $\pre$, \ie{} $\pre \preceq \pre'$, 
and that $\pre$ \emph{under-approximates} $\pre'$, if $\pre(A) \subseteq \pre'(A)$
\end{definition}
\noindent The intention is to measure the over-approximated pre-images of each output event $A$ using the assumed input probability measure $\mu\colon \mathcal{X} \rightarrow [0,1]$, \ie,  requiring $\pre\supsharp(A) \in \mathcal{X}$. 
However, this is not always the case, as shown in the following example.
%However, this is not always the case. %in the following example.
\begin{example}[Examples \ref{ex:f:total+measurable},\ref{ex:f:outputmeasure} continued]\label{ex:f:presharp}
%Recall the measurable input and output space $(\{a,b\},\{\emptyset,\{a,b\}\})$ and $(\{c,d\},\wp(\{c,d\}))$, respectively, and that 
%% $f\colon (\{a,b\},\{\emptyset,\{a,b\}\}) \rightarrow (\{c,d\},\wp(\{b,c\})$ be
%$f\!\colon\! \{a,b\} \! \rightarrow\! \{c,d\}$ is a total and a measurable function
%  whereby $f(a) {=} f(b) {=} c$.
%  
An example of a backwards analysis $\pre\supsharp_f$ that over-approxi\-mates $\pre_f$ could be defined so that $\pre\supsharp_f(\{d\}) =\{b\}$ and $\pre\supsharp_f(\{c\}) =\{a,b\}$. Here, $\{b\}= \pre\supsharp_f(\{d\})\supseteq \pre_f(\{d\}) =\emptyset$ as required, however, $\pre\supsharp_f(\{d\})$ is not measurable in the input space, \ie $\!\{b\} \notin \{\emptyset,\{a,b\}\!\}$. %, and its output probability is undefined since $\mu(\{b\})$.
\end{example}
\noindent For these cases, we define a function  $\uparrow$ that further over-approximates the pre-images.
\begin{definition}
Let $(X,\mathcal{X})$ be a measurable space.
 A function $\uparrow\colon \wp(X) \rightarrow \mathcal{X}$ is an \emph{abstraction}
 if $A \subseteq\; \abs{A}$.
\end{definition}
\noindent An abstraction  $\uparrow\colon \wp(X) \rightarrow \mathcal{X}$ can always be defined using a mapping $f\colon X \rightarrow \mathcal{X}$, where $x \in f(x)$, \ie{} $\abs{B} \triangleq \bigcup_{b\in B}f(b)$. 
The composition of a pre-image over-approximation $\pre\supsharp$ and an abstraction $\uparrow$ over-approximates the pre-image and produces measurable input events.  
\begin{lemma}\label{chp7:lem:comp}\setlength\emergencystretch{.5\textwidth}
Let $\pre\colon \wp(Y) \rightarrow \wp(X)$ be a pre-image, let $\pre \preceq \pre\supsharp$, and let $\uparrow\colon \wp(X) \rightarrow \mathcal{X}$ be an abstraction; then,
$\pre \preceq \; \abs{\;\;\comp\pre\supsharp} \text{ and } \abs{(\pre\supsharp(A))} \in \mathcal{X}.$
\end{lemma}
%\begin{proof} The proof is trivial using the definitions of $\uparrow$ and $\preceq$. \end{proof}  
 %
We can now present the first result, namely the definition of an upper probability bound.
\begin{theorem}\label{chp7:thm:back-up}
Let $f\colon (X,\mathcal{X}) \rightarrow (Y,\mathcal{Y})$ be a measurable function, let $(X,\mathcal{X},\mu)$ be an input probability space, and  let $\mu_f\colon \mathcal{Y} \rightarrow [0,1]$ be the output probability measure. %, \ie{} $\mu_f = \mu \comp \pre_f$.
Furthermore, let 
 $\pre_f\supsharp\colon \wp(Y) \rightarrow \wp(X)$ over-approximate $\pre_f$ and let 
 $\uparrow\colon \wp(X) \rightarrow \mathcal{X}$ be an abstraction. 
 We define the \emph{upper probability bound of $\mu_f$} as
  $\mu_f\supsharp \triangleq \mu \comp \uparrow \comp\!\! \pre\supsharp_f$.
Then,
$ \mu_f(A) \leq \mu_f\supsharp(A)$ and 
%$$ \mu({pre}\supflat_f(A)) = 1 - \mu(\pre\supsharp_f(\stcompc{A})).$$
when $\pre_f\supsharp$ and $\uparrow$ are monotonic, then $\mu_f\supsharp$ is monotonic.
%
%$\mu \comp \pre\supsharp_f$ and $\mu \comp {pre}\supflat_f$ are monotone.
\end{theorem}
\begin{proof}
Lemma~\ref{chp7:lem:comp} yield that $\abs{(\pre\supsharp(A))}\in \mathcal{X}$ and $\pre_f \preceq \;\abs{\,\;\comp\pre\supsharp_f}$. Furthermore, by the monotonicity of $\mu$, we obtain that for any $A \in \mathcal{Y}$ then $\mu(\pre_f(A)) \leq \mu(\abs{\;\;\comp\pre\supsharp_f}(A)$ and thus, $\mu_f(A) \leq \mu_f\supsharp(A)$. 
Finally, when $\uparrow$ and $\pre_f\supsharp$ are monotonic then by composition $\mu \!\comp \uparrow \comp\!\! \pre\supsharp_f$ is monotonic, \ie, $\mu_f\supsharp$ is monotonic, since $\mu$ is monotonic by definition. 
\end{proof}
\begin{example}[Examples \ref{ex:f:total+measurable}–-\ref{ex:f:presharp} continued]
\label{ex:f:presharpabs}
To obtain measurable input events, we create an abstraction $\uparrow_{f}\colon \wp(\{a,b\}) \rightarrow \{\emptyset, \{a,b,\}\}$ defined so that $\uparrow_{f}(\{a\})=\uparrow_{f}(\{b\})= \{a,b\}$ and $\uparrow_{f}(\emptyset)=\uparrow_{f}(\{a,b\}) =id$. According to Theorem~\ref{chp7:thm:back-up},  $\mu_f\supsharp = \mu \comp \uparrow_{f} \comp \pre\supsharp_f$ provides upper probability bounds for the output events as follows --their exact probabilities $\mu_f$ are provided for comparison: 
$$\begin{array}{llll@{\qquad\qquad\big(\;}l@{\;\big)}}
\mu_f\supsharp(\emptyset) &= \mu(\uparrow_{f}(\pre\supsharp_f(\emptyset))) &= \mu(\emptyset) &= 0 & \geq 0 =\mu_f(\emptyset) \\ 
\mu_f\supsharp(\{c\}) &= \mu(\uparrow_{f}(\pre\supsharp_f(\{c\}))) &= \mu(\{a,b\}) &= 1 & \geq 1 =\mu_f(\{c\}) \\
\mu_f\supsharp(\{d\}) &= \mu(\uparrow_{f}(\pre\supsharp_f(\{d\}))) &= \mu(\{a,b\}) &= 1 & \geq 0 =\mu_f(\{d\}) \\
\mu_f\supsharp(\{c,d\}) &= \mu(\uparrow_{f}(\pre\supsharp_f(\{d\}))) &= \mu(\{a,b\}) &= 1 & \geq 1 =\mu_f(\{c,d\})
\end{array}$$
\end{example} 
Based on over-approximating pre-images we also want to derive lower probability bounds; to achieve this we will define under-approximating pre-images and we start by introducing the concept of a dual function.
\begin{definition}%[dual functions]
%The sets $A,B \subseteq X$ are \emph{dual} if $A = \stcompc{B}$.
Let $f, \stcomp{f}$ be functions $f,\stcomp{f}\colon \wp(Y) \rightarrow \wp(X)$. $\stcomp{f}$ is \emph{dual of $f$} if $\stcomp{f}(A) = \stcompc{f(\stcompc{A})}$.
%
%Let $f\colon \wp(Y) \rightarrow \wp(X)$ be a function. A function $\stcomp{f}\colon \wp(Y) \rightarrow \wp(X)$ is \emph{dual of $f$} if $\stcomp{f}(A) = \stcompc{f(\stcompc{A})}$. 
\end{definition}
\noindent Because $\prg$ is total, we can use the \emph{dual} of $\pre\supsharp$ to define a function $\pre\supflat$ that under-approximates $\pre$, as shown by the following lemma.
\begin{lemma}\label{chp7:lem:dualunderapproxpre}
Let $\pre_f,\pre_f\supsharp\!\colon\! \wp(Y) \rightarrow \wp(X)$ be functions with $\pre_f$ as the pre-image of a total and measurable function $f\!\colon\! X \rightarrow Y$ and $\pre_f \! \preceq \!\pre_f\supsharp$. Then, the dual $\pre_f\supflat \triangleq \stcomp{\pre_f\supsharp}$ under-approximates $\pre_f$, \ie, $\pre_f\supflat \preceq \pre_f.$
\end{lemma}
\begin{proof}
Let $A \in \mathcal{Y}$. Since $f$ is total, $\pre_f(A) \cup \pre_f(\stcompc{A}) = X$. Thus, 
$\pre_f\supflat(A) = \stcomp{\pre_f\supsharp} = \stcompc{\pre_f\supsharp(\stcompc{A})}  = \break X\! \setminus \!\pre_f\supsharp(\stcompc{A})  =  \left(\pre_f(A) \cup \pre_f(\stcompc{A}) \right)\!\setminus\! \pre_f\supsharp(\stcompc{A}) 
= \pre_f(A) \setminus \pre_f\supsharp(\stcompc{A}) \subseteq \pre_f(A).
$
%\vspace{-1em} 
%\begin{align*}
%\pre_f\supflat(A) &= \stcompc{\pre_f\supsharp(\stcompc{A})}  = X\! \setminus \!\pre_f\supsharp(\stcompc{A})  = \left(\pre_f(A) \cup \pre_f(\stcompc{A}) \right)\!\setminus\! \pre_f\supsharp(\stcompc{A}) \\
%& = \pre_f(A) \setminus \pre_f\supsharp(\stcompc{A}) \subseteq \pre_f(A).
%\end{align*}
\end{proof}
\begin{lemma}
\label{chp7:prop:prebismonotone}
If $\pre\supsharp_f$ is monotone, then ${pre}\supflat_f$ is monotone.
\end{lemma}
\begin{proof}\label{chp7:prop:prebismonotone:app}
Assume $A,B \subseteq X$, where $A \subseteq B$ and define $C = (B \setminus A)$; then, 
${pre}\supflat_f(B) 
= {pre}\supflat_f(A \uplus C)
= \pre\supsharp_f(A \uplus C) \setminus \pre\supsharp_f(\stcompc{(A \uplus C)})
= \pre\supsharp_f(A \uplus C) \setminus \pre\supsharp_f(\stcompc{A} \cap \stcompc{C}))
\supseteq \pre\supsharp_f(A \uplus C) \setminus \pre\supsharp_f(\stcompc{A})
\supseteq \pre\supsharp_f(A) \setminus \pre\supsharp_f(\stcompc{A})= {pre}\supflat_f(A)
$
\end{proof}
\noindent When the over-approximated pre-images of the output events are measurable in the input measure space, their dual  under-approximated pre-images are also measurable, as the following lemma states.
\begin{restatable}{lemma}{lemdualmeasurable}
%\begin{lemma}
\label{chp7:lem:dualmeasurable}
Let $f\!\colon\! (X, \mathcal{X}) \!\rightarrow \!(Y,\mathcal{Y})$ be a measurable function, 
$\pre_f\supsharp\!\colon \wp(X) \rightarrow \wp(Y)$ be a function where $\pre_f \preceq \pre_f\supsharp$, and $\pre_{f}\supflat$ be dual to $\pre_{f}\supsharp$. Then, for all $A \!\in\! \mathcal{Y}$, 
$\pre_{f}\supsharp(A) \!\in\! \mathcal{X} \text{ if and only if }\pre_{f}\!\supflat(A)\!\in\! \mathcal{X}$
%\end{lemma}
\end{restatable}
\begin{proof}\label{lem:dualmeasurable:app}
Let $A \in \mathcal{Y}$. The following are consequences of $\sigma$-algebras being closed under complements, of  the duality of $\pre_{f}\supflat$ and $\pre\supsharp$, and of the assumed measurability of $A$.  \\ 
``$\Rightarrow$'': 
$A\in \mathcal{Y} \Rightarrow 
\stcompc{A} \in \mathcal{Y} \Rightarrow 
\pre\supsharp_f(\stcompc{A}) \in \mathcal{X} \Rightarrow 
\stcompc{\pre\supsharp_f(\stcompc{A})} \in \mathcal{X} \Rightarrow 
 {pre}\supflat_f(A) \in \mathcal{X}
$ \\
``$\Leftarrow$'': 
$ A \in \mathcal{Y }\Rightarrow 
\stcompc{A} \in \mathcal{Y}\Rightarrow 
 {pre}\supflat_f(\stcompc{A}) \in \mathcal{X} \Rightarrow 
 \stcompc{\pre\supsharp_f(\stcompc{\stcompc{A}})} \in \mathcal{X} \Rightarrow 
 \stcompc{\pre\supsharp_f(A)} \in \mathcal{X} \Rightarrow 
  \pre\supsharp_f(A) \in \mathcal{X} 
 $
\end{proof}
\noindent
We can now define a lower probability bound which is directly related to the upper probability bound.
\begin{theorem}\label{chp7:thm:back-down}
Let $f\colon (X,\mathcal{X}) \rightarrow (Y,\mathcal{Y})$ be a measurable function, let $(X,\mathcal{X},\mu)$ be an input probability space, and  let $\mu_f\colon \mathcal{Y} \rightarrow [0,1]$ be the output probability measure. %, \ie{} $\mu_f = \mu \comp \pre_f$.
Furthermore, let 
 $\pre_f\supsharp\colon \wp(Y) \rightarrow \wp(X)$ over-approximate $\pre_f$ and let 
 $\uparrow\colon \wp(X) \rightarrow \mathcal{X}$ be an abstraction. 
We let ${\pre'}\supflat_f(A) \triangleq \stcompc{\big(\uparrow \comp\!\! \pre\supsharp_f(\stcompc{A})\big)}$ and define the \emph{lower probability bound of $\mu_f$} as
  $\mu\supflat \triangleq \mu \comp {\pre'}\supflat_f$.
  Then,
$  \mu_f\supflat(A) \leq \mu_f(A)$ and $\mu_f\supflat(A) = 1 - \mu_f\supsharp(\stcompc{A})$.
%$$ \mu({pre}\supflat_f(A)) = 1 - \mu(\pre\supsharp_f(\stcompc{A})).$$
Furthermore, if $\pre_f\supsharp$ and $\uparrow$ are monotonic, then $\mu_f\supflat$ is monotonic.
\end{theorem}
\begin{proof}\setlength\emergencystretch{.1\textwidth}
By Lemmas~\ref{chp7:lem:comp} and~\ref{chp7:lem:dualunderapproxpre}, ${\pre'}\supflat_f \preceq \pre_f$, and by~\ref{chp7:lem:comp} and~\ref{chp7:lem:dualmeasurable}, ${\pre'}\supflat_f(A) \in \mathcal{X}$.
Furthermore, by the monotonicity of $\mu$, we obtain $\mu({\pre'}\supflat_f(A)) \leq \mu(\pre_f(A))$ and, thus, $\mu_f\supflat(A) \leq \mu_f(A) \leq \mu_f\supsharp(A)$. 
We obtain the second part by $\mu(A) = 1-\mu(\stcompc{A})$ and the definitions of $\mu_f\supflat$ and $\mu_f\supsharp$, that is,
$\mu_f\supflat(A) = \mu({\pre'}\supflat_f(A)) = 1 - \mu(\stcompc{{\pre'}\supflat_f(A)}) = 
1 - \mu(\stcompc{\stcompc{\big(\uparrow \comp\!\! \pre\supsharp_f(\stcompc{A})\big)}}) = 1 - \mu(\uparrow \comp\!\! \pre\supsharp_f(\stcompc{A}))
= 1 -\mu_f\supsharp(\stcompc{A})$
%
%1 - \mu({\pre'}\supsharp_f(\stcompc{A})) = 1 -\mu_f\supsharp(\stcompc{A})$.
Finally, since $\mu$, $\pre_f\supsharp$ and $\uparrow$ are monotonic, then by composition and Lemma~\ref{chp7:prop:prebismonotone}, $\mu \comp {\pre'}_f\supsharp$ and $\mu \comp {\pre'}_f\supflat$ are monotonic.
\end{proof}
\begin{example}[Examples \ref{ex:f:total+measurable}–-\ref{ex:f:presharpabs} continued]
Folowing Theorem~\ref{chp7:thm:back-down},
we define ${\pre'}\supflat_f(A) \triangleq \stcompc{\big(\uparrow \comp\!\! \pre\supsharp_f(\stcompc{A})\big)}$ such that $\mu\supflat \triangleq \mu \comp {\pre'}\supflat_f$
provides the following lower probability bounds for the output events: 
$$\begin{array}{llll@{\qquad\qquad\big(\;}l@{\;\big)}}
\mu_f\supflat(\emptyset) &= \mu(\stcompc{\uparrow_{f}(\pre\supsharp_f(\stcompc{\emptyset}))}) &= 
\mu(\stcompc{\{a,b\}}) &=  0 & \leq 0 =\mu_f(\emptyset) \\ 
\mu_f\supflat(\{c\}) &= \mu(\stcompc{\uparrow_{f}(\pre\supsharp_f(\stcompc{\{c\}}))}) &= 
\mu(\stcompc{\{a,b\}}) &= 0 & \leq 1 =\mu_f(\{c\}) \\ 
\mu_f\supflat(\{d\}) &= \mu(\stcompc{\uparrow_{f}(\pre\supsharp_f(\stcompc{\{d\}}))}) &= 
\mu(\stcompc{\{a,b\}}) &= 0 & \leq 0 =\mu_f(\{d\}) \\
\mu_f\supflat(\{c,d\}) &= \mu(\stcompc{\uparrow_{f}(\pre\supsharp_f(\stcompc{\{c,d\}}))}) &= 
\mu(\stcompc{\emptyset}) &= 1 & \leq 1 =\mu_f(\{c,d\}) \\
%
%%
%%
%\mu_f\supsharp(\{c\}) &= \mu(\uparrow_{f}(\pre\supsharp_f(\{c\}))) &= \mu(\{a,b\}) &= 1 & \leq 1 =\mu_f(\{c\}) \\
%\mu_f\supsharp(\{d\}) &= \mu(\uparrow_{f}(\pre\supsharp_f(\{d\}))) &= \mu(\{a,b\}) &= 1 & \leq 0 =\mu_f(\{d\}) \\
%\mu_f\supsharp(\{c,d\}) &= \mu(\uparrow_{f}(\pre\supsharp_f(\{d\}))) &= \mu(\{a,b\}) &= 1 & \leq 1 =\mu_f(\{c,d\})
\end{array}$$
\end{example} 
To achieve the tightest probability bounds, the abstraction should return %not only  some increased element that is in the $\sigma$-algebra $\mathcal{X}$ but also 
the least increased element in the $\sigma$-algebra. 
However, in general, such a least element does not exist. 
%\begin{restatable}{lemma}{lemnoleast}\label{chp7:lem:noleast}
\begin{lemma}\label{chp7:lem:noleast}
Let $(X,\mathcal{X})$ be a measurable space, and let $A \in \wp(X)$; there does not always exist a least $B \in \mathcal{X}$ such that $A \subseteq B$.
\end{lemma}
\begin{proof} \label{chp7:lem:noleast:app}
\emph{Proof by counterexample.}
A set $A \subseteq X$Â is \emph{co-countable} if $\stcompc{A}$ is countable. We define 
a $\sigma$-algebra $\mathcal{X}$ to be that generated by the collection of all countable and co-countable subsets of $X$. Note that since each singleton set is countable, they all exist in $\mathcal{X}$. Now, let $A \in \wp(X)$ be uncountable with an uncountable complement $\stcompc{A}$. We will show (by contradiction) that there is no least $B\in\mathcal{X}$ such that $A \subseteq B$. 
Assume that there is a least set $B\in\mathcal{X}$ that contains $A$. Then, $B$ would need to be uncountable, and according to the definition of $\mathcal{X}$, $\stcompc{B}$ would be countable. Since $\stcompc{B}$ is countable and $\stcompc{A}$ is uncountable,  $\stcompc{B} \subset \stcompc{A}$. This is equivalent to $A \subset B$, which causes $B\setminus A$ to contain at least one element; let that element be $x$.
Because $\{x\}$ is a singleton set, $\{x\} \in \mathcal{X}$, and because $\mathcal{X}$ is closed under countable intersection, $B \setminus \{x\} \in \mathcal{X}$. This implies that there is another set, namely, $B \setminus \{x\}$, such that $A \subseteq B \setminus \{x\} \subset B$, and thus $B$ is not the least set in $\mathcal{X}$ that contains $A$ - this contradicts our assumption.
\end{proof}
%\end{restatable}
%\begin{proof} 
%The proof can be found in Appendix~p.~\pageref{chp7:lem:noleast:app}
%\end{proof}
\noindent
When the $\sigma$-algebra is a complete lattice, such a least element does exist. For instance, a power set is both a complete lattice, \eg,~\cite[p.394]{Nielson1999}, and a $\sigma$-algebra~\cite[p.65]{Butler2018}. If the $\sigma$-algebra is a complete lattice, then the abstraction is the identity function, \ie{} $\uparrow = \textit{id}$.
\paragraph{Combining analyses}%\todo{overvej at flytte dette afsnit til senere.}
Two black-box analyses may be combined into a tighter analyses. 
\begin{lemma}\label{chp7:lem:combinepre}
Let 
$\pre_f,\pre_f\supsharp, {\pre'}_f\supsharp \colon \wp(Y) \rightarrow \wp(X)$ be three  functions such that $\pre_f \preceq \pre_f\supsharp $ and $ \pre_f \preceq {\pre'}_f\supsharp$.
Then 
$\pre_f(A) \subseteq \pre_f\supsharp(A) \cap {\pre'}_f\supsharp(A).$
\end{lemma}

\subsection{Forward analysis}\label{sec:forward} % {\raisebox{-0.9ex}{$\scriptscriptstyle\sharp$}}
Again, let $\prg\colon X \rightarrow Y$ be a function, and recall that $\img_{\raisebox{0.9ex}{$\scriptscriptstyle\prg$}}(A) \triangleq \{{\prg}(x) \mid x \in A\}$.
In this section, we present a method for computing upper and lower probability bounds for output events provided a probability measure $\mu\colon \mathcal{X} \rightarrow [0,1]$ over the input $X$ and a forward analysis, that is, a computable over-approximation $\img\supsharp_{\raisebox{0.9ex}{$\scriptscriptstyle\prg$}}\colon \wp(X) \rightarrow \wp(Y)$ of the image-function $\img_{\raisebox{0.9ex}{$\scriptscriptstyle\prg$}}\colon \wp(X) \rightarrow \wp(Y)$, \ie $\img_{\raisebox{0.9ex}{$\scriptscriptstyle\prg$}} \preceq \img\supsharp_{\raisebox{0.9ex}{$\scriptscriptstyle\prg$}}$.
To compute the probability of the events, we only need to define a computable pre-image over-approximating function $\pre_{\raisebox{0.9ex}{$\scriptscriptstyle\prg$}}\supsharp$, and then, we can apply Theorems~\ref{chp7:thm:back-up} and~\ref{chp7:thm:back-down} and obtain Theorem~\ref{chp7:thm:forward}. 

We may define the pre-image function based on the image function 
since the image function on the singletons defines the program semantics.
%Since the image function on the singletons is the program semantics, it may define the pre-image.
%The pre-image function can be defined by the image function on the singletons, since that defines the program semantics.
\begin{lemma} For a function $\!f\!$ with  $\img_f\!\colon \!\wp(X) \rightarrow \wp(Y)$ then
 $\pre_{\!f}(A)  \!=\! \{ x \!\in\! X \mid \img_f(\!\{x\}\!) \cap A \!\neq\! \emptyset\}.$
\end{lemma} 
\begin{proof}
When $f$ is  a function,  $\img_f(\{x\}) = \{f(x)\}$. Thus, the above is a direct consequence of the definition of $\pre_f$, \ie{} $\pre_f(A) \triangleq \{ x \in X \mid f(x) \in A\}$.Â 
\end{proof}
%\begin{proof}
%When $f$ is  a function,  $\img_f(\{x\}) = \{f(x)\}$. Thus, the above is a direct consequence of the definition of $\pre_f$, \ie{} $\pre_f(A) \triangleq \{ x \in X \mid f(x) \in A\}$. 
%\end{proof}
\noindent In our case, we do not have an image function; rather, we have an image over-approximating function $\img\supsharp_f$, \ie $\img_f \preceq \img\supsharp_f$; we may instead use that to define a pre-image over-approximation function $\pre\supsharp_f$. 
\begin{lemma}\label{chp7:lemma:imgftopref1}
Let $f$ be a function with image function
 $\img_f$ and 
 pre-image function 
 $\pre_f$, and let $\img\supsharp_f$ be a function whereby $\img_f \preceq \img\supsharp_f$. If we let $\pre\supsharp_f(A) \triangleq \{ x \in X \mid \img\supsharp_{f}(\{x\}) \cap A \neq \emptyset\}$,
then $\pre\supsharp_f(A) \supseteq \pre_f(A).$
\end{lemma}
\begin{proof}
$\!\!\pre_f\!(A) \!=\! \{x \!\in\! X\!\mid\! \img_f(\{x\}) \!\cap\! A \!\neq\! \emptyset \}\! \subseteq\! \{x {\in} X\!\mid\! \img\supsharp_{f}(\{x\}) {\cap} A \!\neq\! \emptyset \} \!=\! \pre\supsharp_f\!(A).$
\end{proof}
%\begin{proof}
%$\!\!\pre_f\!(A) \!=\! \{x \!\in\! X\!\mid\! \img_f(\{x\}) \!\cap\! A \!\neq\! \emptyset \}\! \subseteq\! \{x {\in} X\!\mid\! \img\supsharp_{f}(\{x\}) {\cap} A \!\neq\! \emptyset \} \!=\! \pre\supsharp_f\!(A).$
%\end{proof}
\noindent For programs whereby $X$ is finite, any output event is computable, but when $X$ is infinite, they are not. %, these pre-images are uncomputable. %, %\ie{} they require infinitely many computations,
% except for the trivial $\sigma$-algebra $\{\emptyset,Y\}$. 
Instead, we propose a computable $\pre\supsharp_f$ based on $\img_f\supsharp$ and a finite partition of input $X$. 
\begin{lemma}
\label{chp7:lem:for}
Let $f\colon (X,\mathcal{X}) \rightarrow (Y,\mathcal{Y})$ be a measurable function, let the function $\img\supsharp_f\colon \wp(X) \rightarrow \wp(Y)$ over-approximate $\img_f$, and let $\mathbf{T}$ denote the set of all partitions over $X$.
We define a function $pre_f\supsharp\colon \mathbf{T} \rightarrow (\wp(Y) \rightarrow \wp(X))$ 
 as $pre_f\supsharp[T](B)  \triangleq \bigcup\{ t \in T \mid \img\supsharp_f(t) \cap B \neq \emptyset \}$
Then, 
$pre_{f} \preceq pre\supsharp_f[T].$
\end{lemma}
\begin{proof}
$pre_{f}(B) =  \{x \in X\mid \img_f(\{x\}) \cap B \neq \emptyset \} 
 \subseteq \{x \in X\mid t \in T \land x \in t \land \img_f(\{t\}) \cap B \neq \emptyset \} 
\subseteq \break \{x \in X\mid t \in T \land x \in t \land \img_f\supsharp(\{t\}) \cap B \neq \emptyset \} 
 \subseteq \bigcup\{ t \in T \mid \img_f\supsharp(\{t\}) \cap B \neq \emptyset \}   = pre\supsharp_f[T](B)
$
\end{proof}
%\begin{proof}
%\begin{align*}
%pre_{f}(B)    &  =  \{x \in X\mid \img_f(\{x\}) \cap B \neq \emptyset \} \\
%              & \subseteq \{x \in X\mid t \in T \land x \in t \land \img_f(\{t\}) \cap B \neq \emptyset \} \\
%              & \subseteq \{x \in X\mid t \in T \land x \in t \land \img_f\supsharp(\{t\}) \cap B \neq \emptyset \} \\
%              & \subseteq \bigcup\{ t \in T \mid \img_f\supsharp(\{t\}) \cap B \neq \emptyset \} \\
%              & = pre\supsharp_f[T](B)
%\end{align*}
%\end{proof}

%\begin{proof} Proof can be found in Appendix p.~\pageref{lem:for:app}
%\end{proof} \vspace{-0.4cm}
%\begin{restatable}{proposition}{proppreclosedunderunion}
%\label{chp7:prop:preclosedunderunion}
%$\pre_f\supsharp[T](\bigcup_{A\in \mathbf{A}}A) = \bigcup_{A\in \mathbf{A}}\pre_f\supsharp[T](A)$.
%\end{restatable}
%\begin{proof} Proof can be found in Appendix p.~\pageref{prop:preclosedunderunion:app}.\end{proof}
%\begin{corollary} \label{chp7:prop:finitemeas}
%For any partition $T$ over $X$, $pre\supsharp_f[T]$ is monotone.
%\end{corollary}
\begin{proposition}
\label{chp7:prop:preclosedunderunion}
$\pre_f\supsharp[T](\bigcup_{A\in \mathbf{A}}A) = \bigcup_{A\in \mathbf{A}}\pre_f\supsharp[T](A)$.
\end{proposition}
\begin{proof}\label{prop:preclosedunderunion:app}
$a \in pre\supsharp_f[T](A \cup B)  
    \; \Leftrightarrow \;
\exists t \in T \colon  a \in t \land \img\supsharp_f(t) \cap (A \cup B) \neq \emptyset
    \; \Leftrightarrow \; 
\exists t \in T \colon  a \in t \land ( \img\supsharp_f(t) \cap A \neq \emptyset)  \;\lor \; ( \img\supsharp_f(t) \cap B \neq \emptyset)
    \; \Leftrightarrow \;
\exists t \in T \colon  (a \in t \land \img\supsharp_f(t) \cap A \neq \emptyset) \;\lor \; (a \in t \land \img\supsharp_f(t) \cap B \neq \emptyset)
    \; \Leftrightarrow \;
a \in pre\supsharp_f[T](A) \cup pre\supsharp_f[T](B)
$
%
%$$
%\begin{array}{lll}
% & & a \in pre\supsharp_f[T](A \cup B) \\
%\Leftrightarrow & \exists t \in T : & a \in t \land \img\supsharp_f(t) \cap (A \cup B) \neq \emptyset\\
%\Leftrightarrow & \exists t \in T : & a \in t \land ( \img\supsharp_f(t) \cap A \neq \emptyset)  \;\lor \; ( \img\supsharp_f(t) \cap B \neq \emptyset) \\
%%
%\Leftrightarrow & \exists t \in T : & (a \in t \land \img\supsharp_f(t) \cap A \neq \emptyset) \;\lor \; (a \in t \land \img\supsharp_f(t) \cap B \neq \emptyset) \\
%%
%\Leftrightarrow &  & a \in pre\supsharp_f[T](A) \cup pre\supsharp_f[T](B)
%\end{array}
%$$
\end{proof}
\begin{corollary} \label{chp7:prop:finitemeas}
For any partition $T$ over $X$, $pre\supsharp_f[T]$ is monotone.
\end{corollary}

\noindent Applying Theorems~\ref{chp7:thm:back-up} and~\ref{chp7:thm:back-down} to the computable $\pre\supsharp[T]$, we obtain computable upper/lower probabilities.
\begin{theorem}\label{chp7:thm:forward}
Let $(X,\mathcal{X},\mu)$ be a probability space, and let $f\colon (X,\mathcal{X}) \rightarrow (Y,\mathcal{Y})$ be a measurable function that induces the output probability measure $\mu_f\colon \mathcal{Y} \rightarrow [0,1]$, \ie{} $\mu_f = \mu \comp \pre_f$. Given a function $\img\supsharp_f\colon \wp(X) \rightarrow \wp(Y)$ such that $\img_f \preceq \img_f\supsharp$,  a finite  partition $T$ over $X$, and a monotone abstraction $\uparrow\colon \wp(X) \rightarrow \mathcal{X}$,
we let ${\pre'}_f\supsharp[T](A) \triangleq \abs{\bigcup\{ t  \mid \exists t\in T\colon \img_f\supsharp(t) \cap A \neq \emptyset \}} $ and 
${\pre'}_f\supflat[T](A) \triangleq \stcompc{{\pre'}\supsharp_f[T](\stcompc{A})}$,
and we define 
$\mu_f\supsharp \triangleq \mu \comp {{\pre'}}\supsharp_f[T]$ and
$\mu_f\supflat \triangleq \mu \comp {\pre'}\supflat_f[T]$.
Then, 
\begin{inparaenum}[(i)]
\item $ \mu_f\supflat(A) \leq \mu_f(A) \leq \mu_f\supsharp(A)$
\item $ \mu_f\supflat(A) = 1 - \mu_f\supsharp(\stcompc{A}), \text{ and}$ \label{chp:thm7:II}
\item $\mu_f\supflat $ and $\mu_f\supsharp$ are monotone.
\end{inparaenum}
\end{theorem}
\begin{proof}
The function $\pre\supsharp_f[T]$ over-approximates $\pre_f$ (by Lemma~\ref{chp7:lem:for}), and it is monotone (by Corollary~\ref{chp7:prop:finitemeas}). Thus, the above is a direct consequence of Theorems~\ref{chp7:thm:back-up} and~\ref{chp7:thm:back-down}. 
\end{proof}
%\paragraph{Choice of partition}
\noindent %The precision of the probabilistic bounds are influenced by the how well the chosen input partition relates to the abstract domain of the black-box analysis. However, 
%Because we assumed no knowledge about the black-box analysis; we may instead exploit the assumed input probability measure with the measurable space $(X,\mathcal{X})$.
%
%
%When choosing the input partition; 
%We may know the analysis we use and its abstract domain and could possibly exploit this knowledge when choosing the input partition; however, this may not be the case. Instead, recall that we assumed that we knew the input probability measure, and thus we  know the input space $(X,\mathcal{X})$. This at least provides us a basis for choosing a partition fitting the input space. 

When choosing the partition $T$ with elements measurable in $\mathcal{X}$, then $\pre\supsharp_f[T](A) \in \mathcal{X}$ for every output event $A$. In this case
we may use identity as abstraction and can unfold the $\mu_f\supsharp$ and $\mu_f\supflat$ into a simpler form. % (see following theorem); this simpler form is used in the case studies. 

%When choosing the partition $T$ such that the elements $t\in T$ are measurable, \ie{} $t\in \mathcal{X}$, then $\pre\supsharp_f[T](A) \in \mathcal{X}$ for every output event $A$;
%this reduces the abstraction $\uparrow$ to the identity function, and we may unfold the $\mu_f\supsharp$ and $\mu_f\supflat$ into a simpler form (see following theorem); this simpler form is used in the case studies. 
\begin{restatable}{theorem}{thmsimpleBounds}
\label{chp7:thm:simpleBounds}
Let $f\colon (X,\mathcal{X}) \rightarrow (Y,\mathcal{Y})$ be a measurable function with the image function $\img_f$ and the pre-image function $\pre_f$, let $(X,\mathcal{X},\mu)$ be a probability space, let $\img\supsharp_f\colon \wp(X) \rightarrow \mathcal{Y}$ be a function that over-approximates $\img_f$, and let
$T$ be a finite partition over $X$ such that $T \subseteq \mathcal{X}$. Then,
$$ \textstyle\mu\supsharp_f(A) =  \sum_{t\in T, \img\supsharp(t) \cap A \neq \emptyset}\mu(t)  \qquad\text{ and }\qquad 
\mu\supflat_f(A) = \sum_{t\in T, \img\supsharp(t) \subseteq A}\mu(t).$$
\end{restatable}
\begin{proof}
Both proof parts rely on Thm.~\ref{chp7:thm:forward} where $\uparrow$ is the identity function, and $T$'s elements are non-overlapping and mesaurable and $\mu$ is  additive; furthermore, the proof of $\mu\supflat_f$ relies on the unfolding %$ \img_f\supsharp(t) \cap \stcompc{A} \neq \emptyset \Leftrightarrow  \img_f\supsharp(t) \subseteq A$, respectively.
${\pre'}_f\supflat[T](A) = {\pre'}\supsharp_f[T](A) \setminus {\pre'}\supsharp_f[T](\stcompc{A})
= \bigcup\{ t\in T \mid \img_f\supsharp(t) \cap A \neq \emptyset \land \neg(\img_f\supsharp(t) \cap \stcompc{A} \neq \emptyset ) \}
= \bigcup\{ t\in T \mid \img_f\supsharp(t) \cap A \neq \emptyset \land \img\supsharp(t) \subseteq A\} = \bigcup\{ t\in T \mid \img\supsharp(t) \subseteq A \}
$.
%
%$\pre_f\supflat(A) = \stcomp{\pre_f\supsharp} = \stcompc{\pre_f\supsharp(\stcompc{A})}  = \break X\! \setminus \!\pre_f\supsharp(\stcompc{A})  =  \left(\pre_f(A) \cup \pre_f(\stcompc{A}) \right)\!\setminus\! \pre_f\supsharp(\stcompc{A}) 
%= \pre_f(A) \setminus \pre_f\supsharp(\stcompc{A}) \subseteq \pre_f(A).
%$
%
\end{proof}
\noindent For readers familiar with Dempster-Shafer theory, the partition is a set of focal elements and Theorem \ref{chp7:thm:forward} generalizes to any set of (overlapping) focal elements; Theorem~\ref{chp7:thm:simpleBounds} resembles the belief and plausibility functions defined based on focal elements~\cite{Klir2007}. Furthermore, the lower probability bounds defines a belief function, which have been related to inner measures~\cite{Ruspini,Halpern1989,Shafer1990}.
%\begin{proof}
%The proof is straight-forward but lengthy; it can be found in Appendix~p.~\pageref{thm:simpleBounds:app}.
%\end{proof}

 %This simpler form is used in the case studies. 
When $\img\supsharp$ is monotone, a finer partition yields tighter probability bounds; however, for a countable infinite $X$, there is no finest finite partition.

\paragraph{Combining analyses}
%\todo{overvej at flytte dette afsnit}
When one or more analyses are forward analyses we can apply the above methods and combine the resulting pre-images using Lemma~\ref{chp7:lem:combinepre}; when both are forward analyses they can be combined directly. 
\begin{lemma}\label{chp7:lem:combineimg}
Let 
$\img_f, \img_f\supsharp, {\img'}_f\supsharp \colon \wp(X) \rightarrow \wp(Y)$ be three functions such that $\img_f \preceq \img_f\supsharp $ and $ \img_f \preceq {\img'}_f\supsharp$.
Then,
$\img_f(A) \subseteq \img_f\supsharp(A) \cap {\img'}_f\supsharp(A).$
\end{lemma}
\section{Examples}\label{sec:casestudies}
In the following we apply the above presented forward approach to two output analysis, namely a sign analysis and an interval analysis and compose them with termination analysis. We study three simple programs; for two of them we provide step-by-step example calculations. For the third program, we first provide the results showing an improvement compared to the essential and still cutting-edge results by Monniaux~\cite{Monniaux2001} and, afterwards, we provide a simple example demonstrating what causes the difference in the results.
\subsection{Sign and Termination analyses}\label{sec:sign}
The program \texttt{sum} (Figure~\ref{chp7:fig:nontermNsign:program:1}) calculates $\sum_{i=1}^{\texttt{x}}i$ for an input \texttt{x}.
We analyse the output properties $\wp(S)$, $S=\{\mathbb{Z}^{-},\{0\},\mathbb{Z}^+, \bot\}$ where $\bot$ represents non-terminating computations. %, that assumes independence of the variables. 
We will derive upper and lower probability bounds for the program's output events 
\begin{enumerate}
\item reusing a standard sign analysis, \eg{},~\cite{Nielson2007}, \ie, a forward analysis $\img^{\sharp}$,
\item reusing the online termination analysis AProVE~\cite{Giesl2017}, \ie, a forward analysis $\img'^{\sharp}$, and 
\item by combining the image-over-approximating functions using Lemma~\ref{chp7:lem:combineimg} constructing a new $\img''^{\sharp}$ that provides results that are stronger than we could from the probability bounds. 
\end{enumerate}

Since they are both forward analysis and the inputs are integers, we use the formulas provided in Theorem~\ref{chp7:thm:simpleBounds}. 
We will analyse the program with respect to input partition $T = \{\mathbb{Z}^{-},\{0\},\mathbb{Z}^+\}$ with input event probabilities as follows: $\mu(\mathbb{Z}^{-})= 1/3$, $\mu(\{0\})=1/4$, and $\mu(\mathbb{Z}^{+})=5/12$.
 
\paragraph{Sign analysis} The sign analysis yields partial correctness \emph{if the program terminates, the analysis' result contains the concrete program result}~\cite{Nielson2007}. Such analyses do (obviously) not conclude anything about termination/non-termination,
 and we safely assume that the output of the program is that of the analysis or $\bot$, see column $\img^{\sharp}$ %_{|\texttt{sum}|,\text{sign}}$
  in 
Figure~\ref{chp7:fig:nontermNsign:table:1}.
Using the formulas from Theorem~\ref{chp7:thm:simpleBounds} we calculate the inferred upper and lower probabilistic bounds for each of the output properties, \eg, see the following example. 
\begin{example}\label{ex:sign:low+up}
To calculate the upper and lower probability of the output event $\{0\}$ the formulas from Theorem~\ref{chp7:thm:simpleBounds} require that we sum the probabilities of the input events in $T$ whose image overlaps with $\{0\}$ and we sum the probabilities of the input events in $T$ whose image is a subset of $\{0\}$, respectively.
$$\begin{array}{lrl}
\mu\supsharp_f(\{0\}) = & \sum_{t\in T, \img\supsharp(t) \cap \{0\} \neq \emptyset}\mu(t) 
 & = \sum_{t\in \{\{0\},\mathbb{Z}^{+}\}}\mu(t) 
 = \mu(\{0\}) + \mu(\mathbb{Z}^{+}) = 1/4 + 5/12 = 2/3 \\
\mu\supflat_f(\{0\}) =& \sum_{t\in T, \img\supsharp(t) \subseteq \{0\}}\mu(t) & = 
\sum_{t\in \emptyset}\mu(t) = 0
\end{array}
$$
The lower probability bound of the output event $\{0\}$ is $0$ and its upper probability bound is $2/3$; in comparison the correct probability of $\{0\}$ is $1/4$.
\end{example}
\noindent The inferred upper and lower probabilistic bounds are shown by blue dashed lines in Figure~\ref{chp7:fig:combination:1}; the correct probabilities are given by orange `$\times$'.
\begin{figure}[h] \centering %\vspace{-0.4cm}
\begin{subfigure}{0.45\textwidth} \centering \small
\begin{Verbatim}[frame= none,baselinestretch=0.1]
int (sum)(int x)
{  int y = 0;
  while (x!= 0){
     y = y + x;
     x = x - 1;}
  return y;     }
  \end{Verbatim}
  \caption{}\label{chp7:fig:nontermNsign:program:1}
\end{subfigure}
\begin{subfigure}{0.45\textwidth}\small
\begin{tabular}{l@{$\;$}|@{$\;$}l@{$\;$}|@{$\;$}l@{$\;$}}
$t \in T$        & $\mu(t)$ &  % $\img^{\raisebox{-0.9ex}{$\scriptscriptstyle\sharp$}}_{|\texttt{sum}|,\text{sign}}$ \\ \hline
 $\img^{\raisebox{-0.9ex}{$\scriptscriptstyle\sharp$}}$ \\ \hline
$\mathbb{Z}^{-}$ & $1/3$  & $\{\bot\} \cup \mathbb{Z}^{-}$  \\
$\{0\}$          & $1/4$  & $\{\bot, 0\}$  \\
$\mathbb{Z}^{+}$ & $5/12$ & $\{\bot, 0\}\cup \mathbb{Z}^{+}$\\
\end{tabular}\vspace{-0.3cm}
\caption{}\label{chp7:fig:nontermNsign:table:1}
\end{subfigure}
\begin{subfigure}{\textwidth}
\begin{tikzpicture}%[scale=1.1]
      \begin{axis}[%
        height=3.2cm,
        align =center, title style={yshift=-10pt,font=\footnotesize},
        title={Probabilities and inferred probability bounds using a sign analysis},
        xtick={-1,0,1,2,3,4,5,6,7,8,9,10,11,12,13,14,15,16},
        scaled ticks=false,
        ytick={0, 0.33,0.67, 1},
        y tick label style={font=\footnotesize},
        yticklabels={$0$,$1/3$,$2/3$,$1$},
        ymin=-0.1,
        ymax=1.1,
        xmin=-0.5,
        xmax=15.5,
        width=0.95\textwidth,
        axis y line*=left,
        axis x line*=bottom,
        line width=0.2mm,
        x tick label style={rotate=-30,anchor=west,font=\footnotesize},
                xticklabels={~,
                $\emptyset$, 
                $ \bot      $,
                $\mathbb{Z}^{-}$, 
                $\{0\}  $,        
                $\mathbb{Z}^{+} $, 
                $\{\!\bot\!\} \cup \mathbb{Z}^{-}$, 
                ${\{\bot, 0\}   }      $, 
                $\{\!\bot\!\} \cup \mathbb{Z}^{+}$, 
                $\mathbb{Z}^{-} \cup \{0\} $,   
                $\mathbb{Z}^{-} \cup \mathbb{Z}^{+}$, 
                $\{0\} \cup \mathbb{Z}^{+}     $,     
                $S\setminus\!\!\mathbb{Z}^{+}$,% ${ \{\bot, 0 \} \cup \mathbb{Z}^{-}} $, 
                %$\{\bot\} \cup \mathbb{Z}^{-} \cup \mathbb{Z}^{+}$, 
                $S\setminus\!\! \{0\}$, %$\{\bot\} \cup \mathbb{Z} \setminus\!\! \{0\}$,
                $S\setminus\!\!\mathbb{Z}^{-}$,%$ {\{ \bot ,0 \} \cup \mathbb{Z}^{+}  }  $,      
                $S\setminus\!\!\{\bot\}$, %$\mathbb{Z}                           $,            
                $S$, %$\{\bot\} \cup \mathbb{Z} $,
                \empty}                             
      ]

% "background" grid lines      
\addplot[very thin,gray,opacity=0.2] coordinates {(-1,1) (16,1)};      
\addplot[very thin,gray,opacity=0.2] coordinates {(-1,0.66666) (16,0.66666)};      
\addplot[very thin,gray,opacity=0.2] coordinates {(-1,0.33333) (16,0.33333)};
\addplot[very thin,gray,opacity=0.2] coordinates {(-1,0) (16,0)};      

      \addplot[ mark options = solid,mark = -, dashed, mycolor!60!white,shift={(0,0)}] coordinates {(0,0) (0,0)};
      \addplot[ mark options = solid,mark = -, dashed, mycolor!60!white,shift={(0,0)}] coordinates {(1,0) (1,1)};
      \addplot[ mark options = solid,mark = -, dashed, mycolor!60!white,shift={(0,0)}] coordinates {(2,0) (2,0.33)};
      \addplot[ mark options = solid,mark = -, dashed, mycolor!60!white,shift={(0,0)}] coordinates {(3,0) (3,0.67)};
      \addplot[ mark options = solid,mark = -, dashed, mycolor!60!white,shift={(0,0)}] coordinates {(4,0) (4,0.42)};
      \addplot[ mark options = solid,mark = -, dashed, mycolor!60!white,shift={(0,0)}] coordinates {(5,0.33) (5,1)};
      \addplot[ mark options = solid,mark = -, dashed, mycolor!60!white,shift={(0,0)}] coordinates {(6,0.25) (6,1)};
      \addplot[ mark options = solid,mark = -, dashed, mycolor!60!white,shift={(0,0)}] coordinates {(7,0) (7,1)};
      \addplot[ mark options = solid,mark = -, dashed, mycolor!60!white,shift={(0,0)}] coordinates {(8,0) (8,1)};
      \addplot[ mark options = solid,mark = -, dashed, mycolor!60!white,shift={(0,0)}] coordinates {(9,0) (9,0.75)};
      \addplot[ mark options = solid,mark = -, dashed, mycolor!60!white,shift={(0,0)}] coordinates {(10,0) (10,0.67)};
      \addplot[ mark options = solid,mark = -, dashed, mycolor!60!white,shift={(0,0)}] coordinates {(11,0.58) (11,1)};
      \addplot[ mark options = solid,mark = -, dashed, mycolor!60!white,shift={(0,0)}] coordinates {(12,0.33) (12,1)};
      \addplot[ mark options = solid,mark = -, dashed, mycolor!60!white,shift={(0,0)}] coordinates {(13,0.67) (13,1)};
      \addplot[ mark options = solid,mark = -, dashed, mycolor!60!white,shift={(0,0)}] coordinates {(14,0) (14,1)};
      \addplot[ mark options = solid,mark = -, dashed, mycolor!60!white,shift={(0,0)}] coordinates {(15,1) (15,1)};   
 
     \addplot[mark=x, mycolor3,shift={(0,0)}] coordinates  {(0,0) (0,0)};
    \addplot[mark=x, mycolor3,shift={(0,0)}] coordinates  {(1,0.33) (1,0.33)};
    \addplot[mark=x, mycolor3,shift={(0,0)}] coordinates  {(2,0) (2,0)};
    \addplot[mark=x, mycolor3,shift={(0,0)}] coordinates  {(3,0.25) (3,0.25)};
    \addplot[mark=x, mycolor3,shift={(0,0)}] coordinates  {(4,0.42) (4,0.42)};
    \addplot[mark=x, mycolor3,shift={(0,0)}] coordinates  {(5,0.33) (5,0.33)};
    \addplot[mark=x, mycolor3,shift={(0,0)}] coordinates  {(6,0.58) (6,0.58)};
    \addplot[mark=x, mycolor3,shift={(0,0)}] coordinates  {(7,0.75) (7,0.75)};
    \addplot[mark=x, mycolor3,shift={(0,0)}] coordinates  {(8,0.67) (8,0.67)};
    \addplot[mark=x, mycolor3,shift={(0,0)}] coordinates  {(9,0.42) (9,0.42)};
    \addplot[mark=x, mycolor3,shift={(0,0)}] coordinates  {(10,0.67) (10,0.67)};
    \addplot[mark=x, mycolor3,shift={(0,0)}] coordinates  {(11,0.58) (11,0.58)};
    \addplot[mark=x, mycolor3,shift={(0,0)}] coordinates  {(12,0.75) (12,0.75)};    
    \addplot[mark=x, mycolor3,shift={(0,0)}] coordinates  {(13,1) (13,1)};
    \addplot[mark=x, mycolor3,shift={(0,0)}] coordinates  {(14,0.67) (14,0.67)};
    \addplot[mark=x, mycolor3,shift={(0,0)}] coordinates  {(15,1) (15,1)};
       \end{axis}
    \end{tikzpicture}\vspace{-0.6cm}%
    \caption{} \label{chp7:fig:combination:1}
\end{subfigure} \vspace{-0.4cm}
\caption{The upper and lower probabilistic bounds inferred using the sign-analysis are shown by blue dashed lines and the correct probabilities are shown by the orange `$\times$'.}\label{chp7:fig:nontermNsign:1}
\end{figure}
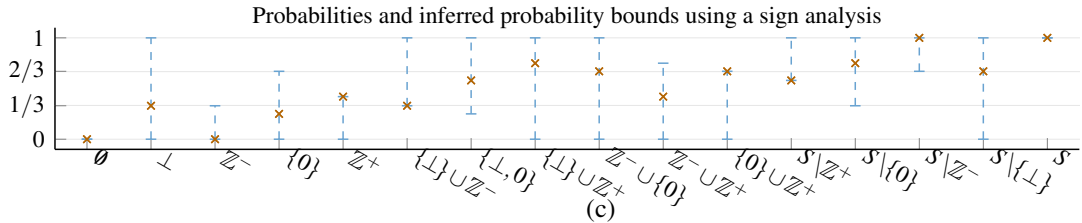
\FloatBarrier
\paragraph{Termination analysis}
%To improve the probability bounds through combination of different analyses, we analysed termination of the program using AProVE~\cite{Giesl2017}. 
We applied the termination analyser AProVE~\cite{Giesl2017} on altered versions of the program\footnote{For each partition element, we made sure that the alternative inputs caused the program to stop and return an integer.} for each partition element $t$ to determine whether $\{\bot\}$ is not in $t$'s image, \ie, $\img'^{\sharp}(t)=\{\mathbb{Z}\}$ or may be a part of $t$'s image, \ie, $\img'^{\sharp}(t)=\{\bot, \mathbb{Z}\}$. The obtained results are displayed in column $\img'^{\sharp}$ in Figure~\ref{chp7:fig:nontermNsign:table:2}. Again, we use the formulas from Theorem~\ref{chp7:thm:simpleBounds} to obtain upper and lower bounds, \eg, see following example.
\begin{example}\label{ex:term:low+up}
To calculate the upper and lower probability of the output event $\{0\}$ the formulas from Theorem~\ref{chp7:thm:simpleBounds} require that we sum the probabilities of the input events in $T$ whose image overlaps with $\{0\}$ and we sum the probabilities of the input events in $T$ whose image is a subset of $\{0\}$, respectively.
Again, we use the formulas from Theorem~\ref{chp7:thm:simpleBounds} to derive upper and lower probability bounds. 
$$\begin{array}{lrl}
\mu\supsharp_f(\{0\}) = & \sum_{t\in T, \img'^{\sharp}(t) \cap \{0\} \neq \emptyset}\mu(t) 
 & = \sum_{t\in \{\mathbb{Z}^{-},\{0\},\mathbb{Z}^{+}\}}\mu(t) 
 = \mu(\mathbb{Z}^{-}) + \mu(\{0\}) + \mu(\mathbb{Z}^{+}) = 1 \\
\mu\supflat_f(\{0\}) =& \sum_{t\in T, \img'^{\sharp}(t) \subseteq \{0\}}\mu(t) & = 
\sum_{t\in \emptyset}\mu(t) = 0
\end{array}
$$
The lower probability bound of the output event $\{0\}$ is $0$ and its upper probability bound is $1$; in comparison the correct probability of $\{0\}$ is $1/4$.
\end{example}
 The inferred probabilistic bounds are shown in Figure~\ref{chp7:fig:combination:2} (green dotted) and are typically worse than those obtained by the sign-analysis.
\begin{figure}[h] \centering \vspace{-0.4cm}
%\begin{subfigure}{0.30\textwidth} \small
%\begin{Verbatim}[frame= none,baselinestretch=0.1]
%int (sum)(int x)
%{  int y = 0;
%  while (x!= 0){
%     y = y + x;
%     x = x - 1;}
%  return y;     }\end{Verbatim}
%  \vspace{-0.6cm}
%  \caption{}\label{chp7:fig:nontermNsign:program}
%\end{subfigure}
\begin{subfigure}{0.22\textwidth}\small
\begin{tabular}{l@{$\;$}|@{$\;$}l@{$\;$}|@{$\;$}l@{$\;$}}
$t \in T$        & $\mu(t)$ &  % $\img^{\raisebox{-0.9ex}{$\scriptscriptstyle\sharp$}}_{|\texttt{sum}|,\text{sign}}$ \\ \hline
 $\img'^{\raisebox{-0.9ex}{$\scriptscriptstyle\sharp$}}$ \\ \hline
$\mathbb{Z}^{-}$ & $1/3$  & $\{\bot\}\cup \mathbb{Z}$ \\
$\{0\}$          & $1/4$  & $\mathbb{Z}$ \\
$\mathbb{Z}^{+}$ & $5/12$ & $\mathbb{Z}$\\
\end{tabular}
%
%\begin{tabular}{l@{$\;$}|@{$\;$}l@{$\;$}|@{$\;$}l@{$\;$}|@{$\;$}l@{$\;$}}
%$t \in T$        &  $\img^{\raisebox{-0.9ex}{$\scriptscriptstyle\sharp$}}_{|\texttt{sum}|,\text{sign}}$ &  $\img^{\raisebox{-0.9ex}{$\scriptscriptstyle\sharp$}}_{|\texttt{sum}|,\text{term}}$  &  $\img^{\raisebox{-0.9ex}{$\scriptscriptstyle\sharp$}}_{|\texttt{sum}|}$ \\ \hline
%$\mathbb{Z}^{-}$ & $\{\bot\} \cup \mathbb{Z}^{-}$  & $\{\bot\}\cup \mathbb{Z}$ & $\{\bot\}\cup \mathbb{Z}^{-}$ \\
%$\{0\}$          &  $\{\bot\} \cup\{0\}$       & $\mathbb{Z}$  & $\{0\}$\\
%$\mathbb{Z}^{+}$ & $\{\bot\} \cup\{0\}\cup \mathbb{Z}^{+}$ & $\mathbb{Z}$ & $\{0\}\cup \mathbb{Z}^{+}$\\
%\end{tabular}\vspace{-0.3cm}
\caption{}\label{chp7:fig:nontermNsign:table:2}
\end{subfigure}~
\begin{subfigure}{0.72\textwidth}
\begin{tikzpicture}%[scale=1.1]
      \begin{axis}[%
        height=3.2cm,
        align =center, title style={yshift=-10pt,font=\footnotesize},
        title={Probabilities and inferred probability bounds using a termination analysis},
        xtick={-1,0,1,2,3,4,5,6,7,8,9,10,11,12,13,14,15,16},
        scaled ticks=false,
        ytick={0, 0.33,0.67, 1},
        y tick label style={font=\footnotesize},
        yticklabels={$0$,$1/3$,$2/3$,$1$},
        ymin=-0.1,
        ymax=1.1,
        xmin=-0.5,
        xmax=15.5,
        width=1.1\textwidth,
        axis y line*=left,
        axis x line*=bottom,
        line width=0.2mm,
        x tick label style={rotate=-30,anchor=west,font=\footnotesize},
                xticklabels={~,
                $\emptyset$, 
                $ \bot      $,
                $\mathbb{Z}^{-}$, 
                $\{0\}  $,        
                $\mathbb{Z}^{+} $, 
                $\{\!\bot\!\} \cup \mathbb{Z}^{-}$, 
                ${\{\bot, 0\}   }      $, 
                $\{\!\bot\!\} \cup \mathbb{Z}^{+}$, 
                $\mathbb{Z}^{-} \cup \{0\} $,   
                $\mathbb{Z}^{-} \cup \mathbb{Z}^{+}$, 
                $\{0\} \cup \mathbb{Z}^{+}     $,     
                $S\setminus\!\!\mathbb{Z}^{+}$,% ${ \{\bot, 0 \} \cup \mathbb{Z}^{-}} $, 
                %$\{\bot\} \cup \mathbb{Z}^{-} \cup \mathbb{Z}^{+}$, 
                $S\setminus\!\! \{0\}$, %$\{\bot\} \cup \mathbb{Z} \setminus\!\! \{0\}$,
                $S\setminus\!\!\mathbb{Z}^{-}$,%$ {\{ \bot ,0 \} \cup \mathbb{Z}^{+}  }  $,      
                $S\setminus\!\!\{\bot\}$, %$\mathbb{Z}                           $,            
                $S$, %$\{\bot\} \cup \mathbb{Z} $,
                \empty}                             
      ]  
 
 % "background" grid lines      
\addplot[very thin,gray,opacity=0.2] coordinates {(-1,1) (16,1)};      
\addplot[very thin,gray,opacity=0.2] coordinates {(-1,0.66666) (16,0.66666)};      
\addplot[very thin,gray,opacity=0.2] coordinates {(-1,0.33333) (16,0.33333)};
\addplot[very thin,gray,opacity=0.2] coordinates {(-1,0) (16,0)};

\addplot[dotted,mark = -, mark options = solid, mycolor2,shift={(0,0)}] coordinates  {(0,0) (0,0)};
    \addplot[densely dotted,mark = -, mark options = solid, mycolor2,shift={(0,0)}] coordinates  {(1,0) (1,0.33)};
    \addplot[densely dotted,mark = -, mark options = solid, mycolor2,shift={(0,0)}] coordinates  {(2,0) (2,1)};
    \addplot[densely dotted,mark = -, mark options = solid, mycolor2,shift={(0,0)}] coordinates  {(3,0) (3,1)};
    \addplot[densely dotted,mark = -, mark options = solid, mycolor2,shift={(0,0)}] coordinates  {(4,0) (4,1)};
    \addplot[dotted,mark = -, mark options = solid, mycolor2,shift={(0,0)}] coordinates  {(5,0) (5,1)};
    \addplot[dotted,mark = -, mark options = solid, mycolor2,shift={(0,0)}] coordinates  {(6,0) (6,1)};
    \addplot[dotted,mark = -, mark options = solid, mycolor2,shift={(0,0)}] coordinates  {(7,0) (7,1)};
    \addplot[dotted,mark = -, mark options = solid, mycolor2,shift={(0,0)}] coordinates  {(8,0) (8,1)};
    \addplot[dotted,mark = -, mark options = solid, mycolor2,shift={(0,0)}] coordinates  {(9,0) (9,1)};
    \addplot[dotted,mark = -, mark options = solid, mycolor2,shift={(0,0)}] coordinates  {(10,0) (10,1)};
    \addplot[dotted,mark = -, mark options = solid, mycolor2,shift={(0,0)}] coordinates  {(11,0) (11,1)};
    \addplot[dotted,mark = -, mark options = solid, mycolor2,shift={(0,0)}] coordinates  {(12,0) (12,1)};    
    \addplot[dotted,mark = -, mark options = solid, mycolor2,shift={(0,0)}] coordinates  {(13,0) (13,1)};
    \addplot[dotted,mark = -, mark options = solid, mycolor2,shift={(0,0)}] coordinates  {(14,0.67) (14,1)};
    \addplot[dotted,mark = -, mark options = solid, mycolor2,shift={(0,0)}] coordinates  {(15,1) (15,1)};

    \addplot[mark=x, mycolor3,shift={(0,0)}] coordinates  {(0,0) (0,0)};
    \addplot[mark=x, mycolor3,shift={(0,0)}] coordinates  {(1,0.33) (1,0.33)};
    \addplot[mark=x, mycolor3,shift={(0,0)}] coordinates  {(2,0) (2,0)};
    \addplot[mark=x, mycolor3,shift={(0,0)}] coordinates  {(3,0.25) (3,0.25)};
    \addplot[mark=x, mycolor3,shift={(0,0)}] coordinates  {(4,0.42) (4,0.42)};
    \addplot[mark=x, mycolor3,shift={(0,0)}] coordinates  {(5,0.33) (5,0.33)};
    \addplot[mark=x, mycolor3,shift={(0,0)}] coordinates  {(6,0.58) (6,0.58)};
    \addplot[mark=x, mycolor3,shift={(0,0)}] coordinates  {(7,0.75) (7,0.75)};
    \addplot[mark=x, mycolor3,shift={(0,0)}] coordinates  {(8,0.67) (8,0.67)};
    \addplot[mark=x, mycolor3,shift={(0,0)}] coordinates  {(9,0.42) (9,0.42)};
    \addplot[mark=x, mycolor3,shift={(0,0)}] coordinates  {(10,0.67) (10,0.67)};
    \addplot[mark=x, mycolor3,shift={(0,0)}] coordinates  {(11,0.58) (11,0.58)};
    \addplot[mark=x, mycolor3,shift={(0,0)}] coordinates  {(12,0.75) (12,0.75)};    
    \addplot[mark=x, mycolor3,shift={(0,0)}] coordinates  {(13,1) (13,1)};
    \addplot[mark=x, mycolor3,shift={(0,0)}] coordinates  {(14,0.67) (14,0.67)};
    \addplot[mark=x, mycolor3,shift={(0,0)}] coordinates  {(15,1) (15,1)};
       \end{axis}
    \end{tikzpicture}\vspace{-0.6cm}%
    \caption{} \label{chp7:fig:combination:2}
\end{subfigure} 
%\end{minipage}
\caption{The upper and lower probabilistic bounds inferred using the termination-analysis are shown by green dotted lines and the correct probabilities are shown by the orange `$\times$'. }\vspace{-0.1cm}
\label{chp7:fig:nontermNsign:2}
\end{figure}
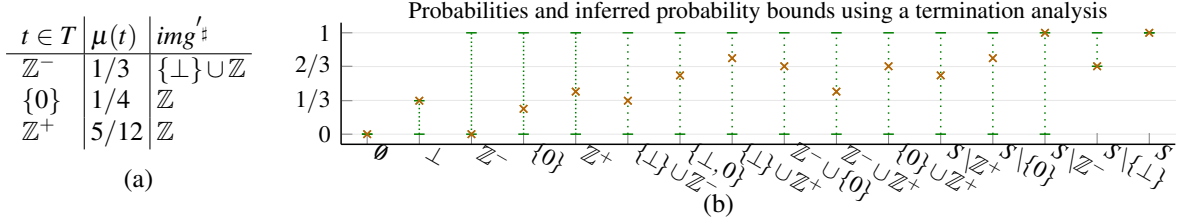
\FloatBarrier
\paragraph{Combined analyses.}
Instead of combining the resulting probabilistic bounds from the two analyses, we have proposed to combine the analyses image-overapproximating functions using Lemma~\ref{chp7:lem:combineimg} and afterwards infer the probability bounds based on this combined result using, \eg, the formulas of Theorem~\ref{chp7:thm:simpleBounds}. %the results from combining the sign analysis and the termination analysis are displayed in column $\img''^{\sharp}$ %$\img^{\raisebox{-0.9ex}{$\scriptscriptstyle\sharp$}}_{|\texttt{sum}|}$
 %in Table~\ref{chp7:fig:nontermNsign:table}. 
% \begin{example}
Given two image-over-approximating functions $\img^{\sharp}$ (see Table~\ref{chp7:fig:nontermNsign:table:1}) and $\img'^{\sharp}$ (see Table~\ref{chp7:fig:nontermNsign:table:2}), we use the formula of Lemma~\ref{chp7:lem:combineimg} to define $\img''^{\sharp}$ over the partition elements as follows.
 $$
 \begin{array}{lclr@{\;}ll}
 \img''^{\sharp}(\mathbb{Z}^{-}) = & \img^{\sharp}(\mathbb{Z}^{-}) \cap   \img'^{\sharp}(\mathbb{Z}^{-}) & = & \left(\{\bot\} \cup \mathbb{Z}^{-} \right) & \cap  \;
 \left( \{\bot\}\cup \mathbb{Z} \right) & = \{\bot\} \cup \mathbb{Z}^{-} \\
  \img''^{\sharp}(\{0\}) = & \img^{\sharp}(\{0\}) \cap   \img'^{\sharp}(\{0\}) & = & \{\bot,0\} & \cap \; \mathbb{Z} & = \{0\} \\
  \img''^{\sharp}(\mathbb{Z}^{+}) = & \img^{\sharp}(\mathbb{Z}^{+}) \cap   \img'^{\sharp}(\mathbb{Z}^{+}) & = &\left(\{\bot,0\} \cup \mathbb{Z}^{+} \right) &\cap \; 
\mathbb{Z} & = \{0\} \cup \mathbb{Z}^{+} \\
 \end{array}$$
%$$
%\begin{tabular}{l@{$\;$}|@{$\;$}l@{$\;$}|@{$\;$}l@{$\;$}|@{$\;$}l@{$\;$}}
%$t \in T$        &  $\img^{\raisebox{-0.9ex}{$\scriptscriptstyle\sharp$}}_{|\texttt{sum}|,\text{sign}}$ &  $\img^{\raisebox{-0.9ex}{$\scriptscriptstyle\sharp$}}_{|\texttt{sum}|,\text{term}}$  &  $\img^{\raisebox{-0.9ex}{$\scriptscriptstyle\sharp$}}_{|\texttt{sum}|}$ \\ \hline
%$\mathbb{Z}^{-}$ & $\{\bot\} \cup \mathbb{Z}^{-}$  & $\{\bot\}\cup \mathbb{Z}$ & $\{\bot\}\cup \mathbb{Z}^{-}$ \\
%$\{0\}$          &  $\{\bot\} \cup\{0\}$       & $\mathbb{Z}$  & $\{0\}$\\
%$\mathbb{Z}^{+}$ & $\{\bot\} \cup\{0\}\cup \mathbb{Z}^{+}$ & $\mathbb{Z}$ & $\{0\}\cup \mathbb{Z}^{+}$\\
%\end{tabular}\vspace{-0.3cm}
% \end{example}
Again, we infer probability bounds based on the formulas of Theorem~\ref{chp7:thm:simpleBounds} using the new image-over\-approximating function $\img''^{\sharp}$, e.g., see the following example.
\begin{example}\label{ex:sign+term:low+up}
The procedure is similar to those in Examples~\ref{ex:sign:low+up} and~\ref{ex:term:low+up}; here, we use the function $\img''^{\sharp}$ to derive upper and lower probability bounds of the output event $\{0\}$. 
$$\begin{array}{lrl}
\mu\supsharp_{|\texttt{sum}|}(\{0\}) = & \sum_{t\in T, \img''^{\sharp}(t) \cap \{0\} \neq \emptyset}\mu(t) 
 & = \sum_{t\in \{\{0\},\mathbb{Z}^{+}\}}\mu(t) 
 =  \mu(\{0\}) + \mu(\mathbb{Z}^{+}) = 1/4 + 5/12 = 2/3 \\
\mu\supflat_{|\texttt{sum}|}(\{0\}) =& \sum_{t\in T, \img''^{\sharp}(t) \subseteq \{0\}}\mu(t) & = 
\sum_{t\in \{\{0\}\}}\mu(t) = \mu(\{0\}) = 1/4
\end{array}
$$
The lower probability bound of the output event $\{0\}$ is $1/4$ and its upper probability bound is $2/3$; in comparison the correct probability of $\{0\}$ is $1/4$. Note that both the previous probabilistic analyses produced 0 as the lower probability bound.
\end{example}
The improved probability bounds are displayed in Figure~\ref{chp7:fig:combination} (solid black) together with the previous results. These combined results are more precise than if we had simply used the minimum and maximum of the bounds of the individual analyses; for instance the upper and lower probability bounds of $\{0\}$ would have been $2/3$ and $0$ and with the suggested method the lower bound is improved to $1/3$. As expected from part~\ref{chp:thm7:II} of Theorem~\ref{chp7:thm:forward}, \ie, $ \mu_f\supflat(A) = 1 - \mu_f\supsharp(\stcompc{A})$, the improvement of the lower probability bound of $\{0\}$ influences the upper probability bound of its complement set, \ie, $S\setminus\!\{0\}$, \ie, $\mu_{|\texttt{sum}|}(S\setminus\!\{0\})$ is reduced from $1$ to $2/3$.

\begin{figure}[h] \centering \vspace{-0.4cm}
%\begin{subfigure}{0.30\textwidth} \small
%\begin{Verbatim}[frame= none,baselinestretch=0.1]
%int (sum)(int x)
%{  int y = 0;
%  while (x!= 0){
%     y = y + x;
%     x = x - 1;}
%  return y;     }\end{Verbatim}
%  \vspace{-0.6cm}
%  \caption{}\label{chp7:fig:nontermNsign:program}
%\end{subfigure}
%\begin{subfigure}{0.60\textwidth}\small
%\begin{tabular}{l@{$\;$}|@{$\;$}l@{$\;$}|@{$\;$}l@{$\;$}|@{$\;$}l@{$\;$}}
%$t \in T$        &  $\img^{\raisebox{-0.9ex}{$\scriptscriptstyle\sharp$}}_{|\texttt{sum}|,\text{sign}}$ &  $\img^{\raisebox{-0.9ex}{$\scriptscriptstyle\sharp$}}_{|\texttt{sum}|,\text{term}}$  &  $\img^{\raisebox{-0.9ex}{$\scriptscriptstyle\sharp$}}_{|\texttt{sum}|}$ \\ \hline
%$\mathbb{Z}^{-}$ & $\{\bot\} \cup \mathbb{Z}^{-}$  & $\{\bot\}\cup \mathbb{Z}$ & $\{\bot\}\cup \mathbb{Z}^{-}$ \\
%$\{0\}$          &  $\{\bot\} \cup\{0\}$       & $\mathbb{Z}$  & $\{0\}$\\
%$\mathbb{Z}^{+}$ & $\{\bot\} \cup\{0\}\cup \mathbb{Z}^{+}$ & $\mathbb{Z}$ & $\{0\}\cup \mathbb{Z}^{+}$\\
%\end{tabular}\vspace{-0.3cm}
%\caption{}\label{chp7:fig:nontermNsign:table}
%\end{subfigure}\vspace{-0.3cm}
%\begin{subfigure}{\textwidth}
\begin{tikzpicture}%[scale=1.1]
      \begin{axis}[%
        height=5 cm,
        align =center, title style={yshift=-10pt,font=\footnotesize},
        title={Probabilities and inferred probability bounds},
        xtick={-1,0,1,2,3,4,5,6,7,8,9,10,11,12,13,14,15,16},
        scaled ticks=false,
        ytick={0, 0.33,0.67, 1},
        y tick label style={font=\footnotesize},
        yticklabels={$0$,$1/3$,$2/3$,$1$},
        ymin=-0.1,
        ymax=1.1,
        xmin=-0.5,
        xmax=15.5,
        width=0.95\textwidth,
        axis y line*=left,
        axis x line*=bottom,
        line width=0.2mm,
        x tick label style={rotate=-30,anchor=west,font=\footnotesize},
                xticklabels={~,
                $\emptyset$, 
                $ \bot      $,
                $\mathbb{Z}^{-}$, 
                $\{0\}  $,        
                $\mathbb{Z}^{+} $, 
                $\{\!\bot\!\} \cup \mathbb{Z}^{-}$, 
                ${\{\bot, 0\}   }      $, 
                $\{\!\bot\!\} \cup \mathbb{Z}^{+}$, 
                $\mathbb{Z}^{-} \cup \{0\} $,   
                $\mathbb{Z}^{-} \cup \mathbb{Z}^{+}$, 
                $\{0\} \cup \mathbb{Z}^{+}     $,     
                $S\setminus\!\!\mathbb{Z}^{+}$,% ${ \{\bot, 0 \} \cup \mathbb{Z}^{-}} $, 
                %$\{\bot\} \cup \mathbb{Z}^{-} \cup \mathbb{Z}^{+}$, 
                $S\setminus\!\! \{0\}$, %$\{\bot\} \cup \mathbb{Z} \setminus\!\! \{0\}$,
                $S\setminus\!\!\mathbb{Z}^{-}$,%$ {\{ \bot ,0 \} \cup \mathbb{Z}^{+}  }  $,      
                $S\setminus\!\!\{\bot\}$, %$\mathbb{Z}                           $,            
                $S$, %$\{\bot\} \cup \mathbb{Z} $,
                \empty}                             
      ]
% "background" grid lines      
\addplot[very thin,gray,opacity=0.2] coordinates {(-1,1) (16,1)};      
\addplot[very thin,gray,opacity=0.2] coordinates {(-1,0.66666) (16,0.66666)};      
\addplot[very thin,gray,opacity=0.2] coordinates {(-1,0.33333) (16,0.33333)};
\addplot[very thin,gray,opacity=0.2] coordinates {(-1,0) (16,0)};

      \addplot[ mark options = solid,mark = -, dashed, mycolor!60!white,shift={(-3,0)}] coordinates {(0,0) (0,0)};
      \addplot[ mark options = solid,mark = -, dashed, mycolor!60!white,shift={(-3,0)}] coordinates {(1,0) (1,1)};
      \addplot[ mark options = solid,mark = -, dashed, mycolor!60!white,shift={(-3,0)}] coordinates {(2,0) (2,0.33)};
      \addplot[ mark options = solid,mark = -, dashed, mycolor!60!white,shift={(-3,0)}] coordinates {(3,0) (3,0.67)};
      \addplot[ mark options = solid,mark = -, dashed, mycolor!60!white,shift={(-3,0)}] coordinates {(4,0) (4,0.42)};
      \addplot[ mark options = solid,mark = -, dashed, mycolor!60!white,shift={(-3,0)}] coordinates {(5,0.33) (5,1)};
      \addplot[ mark options = solid,mark = -, dashed, mycolor!60!white,shift={(-3,0)}] coordinates {(6,0.25) (6,1)};
      \addplot[ mark options = solid,mark = -, dashed, mycolor!60!white,shift={(-3,0)}] coordinates {(7,0) (7,1)};
      \addplot[ mark options = solid,mark = -, dashed, mycolor!60!white,shift={(-3,0)}] coordinates {(8,0) (8,1)};
      \addplot[ mark options = solid,mark = -, dashed, mycolor!60!white,shift={(-3,0)}] coordinates {(9,0) (9,0.75)};
      \addplot[ mark options = solid,mark = -, dashed, mycolor!60!white,shift={(-3,0)}] coordinates {(10,0) (10,0.67)};
      \addplot[ mark options = solid,mark = -, dashed, mycolor!60!white,shift={(-3,0)}] coordinates {(11,0.58) (11,1)};
      \addplot[ mark options = solid,mark = -, dashed, mycolor!60!white,shift={(-3,0)}] coordinates {(12,0.33) (12,1)};
      \addplot[ mark options = solid,mark = -, dashed, mycolor!60!white,shift={(-3,0)}] coordinates {(13,0.67) (13,1)};
      \addplot[ mark options = solid,mark = -, dashed, mycolor!60!white,shift={(-3,0)}] coordinates {(14,0) (14,1)};
      \addplot[ mark options = solid,mark = -, dashed, mycolor!60!white,shift={(-3,0)}] coordinates {(15,1) (15,1)};   
 
 \addplot[dotted,mark = -, mark options = solid, mycolor2,shift={(-1.5,0)}] coordinates  {(0,0) (0,0)};
    \addplot[densely dotted,mark = -, mark options = solid, mycolor2,shift={(-1.5,0)}] coordinates  {(1,0) (1,0.33)};
    \addplot[densely dotted,mark = -, mark options = solid, mycolor2,shift={(-1.5,0)}] coordinates  {(2,0) (2,1)};
    \addplot[densely dotted,mark = -, mark options = solid, mycolor2,shift={(-1.5,0)}] coordinates  {(3,0) (3,1)};
    \addplot[densely dotted,mark = -, mark options = solid, mycolor2,shift={(-1.5,0)}] coordinates  {(4,0) (4,1)};
    \addplot[dotted,mark = -, mark options = solid, mycolor2,shift={(-1.5,0)}] coordinates  {(5,0) (5,1)};
    \addplot[dotted,mark = -, mark options = solid, mycolor2,shift={(-1.5,0)}] coordinates  {(6,0) (6,1)};
    \addplot[dotted,mark = -, mark options = solid, mycolor2,shift={(-1.5,0)}] coordinates  {(7,0) (7,1)};
    \addplot[dotted,mark = -, mark options = solid, mycolor2,shift={(-1.5,0)}] coordinates  {(8,0) (8,1)};
    \addplot[dotted,mark = -, mark options = solid, mycolor2,shift={(-1.5,0)}] coordinates  {(9,0) (9,1)};
    \addplot[dotted,mark = -, mark options = solid, mycolor2,shift={(-1.5,0)}] coordinates  {(10,0) (10,1)};
    \addplot[dotted,mark = -, mark options = solid, mycolor2,shift={(-1.5,0)}] coordinates  {(11,0) (11,1)};
    \addplot[dotted,mark = -, mark options = solid, mycolor2,shift={(-1.5,0)}] coordinates  {(12,0) (12,1)};    
    \addplot[dotted,mark = -, mark options = solid, mycolor2,shift={(-1.5,0)}] coordinates  {(13,0) (13,1)};
    \addplot[dotted,mark = -, mark options = solid, mycolor2,shift={(-1.5,0)}] coordinates  {(14,0.67) (14,1)};
    \addplot[dotted,mark = -, mark options = solid, mycolor2,shift={(-1.5,0)}] coordinates  {(15,1) (15,1)};
 
      \addplot[solid,mark = -, mark options = solid, black,shift={(0,0)}] coordinates  {(0,0) (0,0)};
      \addplot[solid,mark = -, mark options = solid,  black,shift={(0,0)}] coordinates {(1,0) (1, 0.33)};
      \addplot[solid,mark = -, mark options = solid,  black,shift={(0,0)}] coordinates  {(2, 0) (2, 0.33)};
      \addplot[solid,mark = -, mark options = solid,  black,shift={(0,0)}] coordinates  {(3, 0.25) (3, 0.67)}; 
      \addplot[solid,mark = -, mark options = solid,  black,shift={(0,0)}] coordinates  {(4, 0)    (4, 0.42)};
      \addplot[solid,mark = -, mark options = solid,  black,shift={(0,0)}] coordinates  {(5, 0.33) (5, 0.33)};
      \addplot[solid,mark = -, mark options = solid,  black,shift={(0,0)}] coordinates  {(6, 0.25) (6, 1)};
      \addplot[solid,mark = -, mark options = solid,  black,shift={(0,0)}] coordinates  {(7, 0) (7,0.75)};
      \addplot[solid,mark = -, mark options = solid,  black,shift={(0,0)}] coordinates  {(8, 0) (8, 1)};
      \addplot[solid,mark = -, mark options = solid,  black,shift={(0,0)}] coordinates {(9,0) (9, 0.75)};
      \addplot[solid,mark = -, mark options = solid,  black,shift={(0,0)}] coordinates {(10, 0.67) (10, 0.67)};
      \addplot[solid,mark = -, mark options = solid,  black,shift={(0,0)}] coordinates {(11, 0.58) (11, 1)};
      \addplot[solid,mark = -, mark options = solid,  black,shift={(0,0)}] coordinates {(12, 0.33) (12, 0.75)};
      \addplot[solid,mark = -, mark options = solid,  black,shift={(0,0)}] coordinates {(13, 0.67) (13, 1)};
      \addplot[solid,mark = -, mark options = solid,  black,shift={(0,0)}] coordinates {(14, 0.67) (14, 1)};
      \addplot[solid,mark = -, mark options = solid,  black,shift={(0,0)}] coordinates {(15,1) (15,1)};

    \addplot[mark=x, mycolor3,shift={(0,0)}] coordinates  {(0,0) (0,0)};
    \addplot[mark=x, mycolor3,shift={(0,0)}] coordinates  {(1,0.33) (1,0.33)};
    \addplot[mark=x, mycolor3,shift={(0,0)}] coordinates  {(2,0) (2,0)};
    \addplot[mark=x, mycolor3,shift={(0,0)}] coordinates  {(3,0.25) (3,0.25)};
    \addplot[mark=x, mycolor3,shift={(0,0)}] coordinates  {(4,0.42) (4,0.42)};
    \addplot[mark=x, mycolor3,shift={(0,0)}] coordinates  {(5,0.33) (5,0.33)};
    \addplot[mark=x, mycolor3,shift={(0,0)}] coordinates  {(6,0.58) (6,0.58)};
    \addplot[mark=x, mycolor3,shift={(0,0)}] coordinates  {(7,0.75) (7,0.75)};
    \addplot[mark=x, mycolor3,shift={(0,0)}] coordinates  {(8,0.67) (8,0.67)};
    \addplot[mark=x, mycolor3,shift={(0,0)}] coordinates  {(9,0.42) (9,0.42)};
    \addplot[mark=x, mycolor3,shift={(0,0)}] coordinates  {(10,0.67) (10,0.67)};
    \addplot[mark=x, mycolor3,shift={(0,0)}] coordinates  {(11,0.58) (11,0.58)};
    \addplot[mark=x, mycolor3,shift={(0,0)}] coordinates  {(12,0.75) (12,0.75)};    
    \addplot[mark=x, mycolor3,shift={(0,0)}] coordinates  {(13,1) (13,1)};
    \addplot[mark=x, mycolor3,shift={(0,0)}] coordinates  {(14,0.67) (14,0.67)};
    \addplot[mark=x, mycolor3,shift={(0,0)}] coordinates  {(15,1) (15,1)};
       \end{axis}
    \end{tikzpicture}\vspace{-0.6cm}%
    \caption{} \label{chp7:fig:combination}
%\end{subfigure} \vspace{-0.4cm}
%\end{minipage}
\caption{Probability bounds inferred by $\img^\sharp$ (using the sign-analysis) are shown as blue dashed lines, those inferred by $\img'^\sharp$ (using the termination-analysis) are shown as green dotted lines, and those inferred by their combined image-overapproximating function $\img''^\sharp$ are shown as solid black lines. For comparison, the correct probability distribution is indicated by orange `$\times$'.}
\label{chp7:fig:nontermNsign}
\end{figure}
\FloatBarrier
\subsection{Interval analyses}\label{sec:casestudyinterval}
%\section{Case study: Interval analyses}\label{sec:casestudyinterval}
%\todo[inline]{
%In the introduction, the author mentions that Monniaux'
%results are "more general". While making full use of this generality
%requires the development of novel abstract domains, it is not made
%sufficiently clear whether Monniaux framework could be *specialized*
%to obtain the results of the paper
%
%In the case study section, the experimental results show an
%advantage over Monniaux' analysis. However, the advantage is gradual
%and seems largely unrelated to the question of probabilistic vs
%nonprobabilistic analysis.}
We study the following simple programs \texttt{f} and \texttt{g} and compare our results with those by Monniaux; the first program \texttt{f} has no branching points, but the second program \texttt{g} do.\vspace{0.2cm} \\ \vspace{0.2cm} 
\begin{minipage}{0.45\textwidth}\small
\begin{Verbatim}[frame=none,baselinestretch=1]
 double f(double x1, x2, x3, x4){ 
   double x; x = 0.0;
   x = x+ x1*2.0-1.0;
   x = x+ x2*2.0-1.0;
   x = x+ x3*2.0-1.0;
   x = x+ x4*2.0-1.0;
   return x; }    
\end{Verbatim}
\end{minipage}
\begin{minipage}{0.52\textwidth}\small
%double orgg() {
%  double x=0.0;
%  if(drand48()>0.5)
%    x += drand48()*2.0-1.0;
%  for (int i=0; i<4; i++) 
%    x += drand48()*2.0-1.0;
%  return x; } 
%  
\begin{Verbatim}[frame=none,baselinestretch=1]
double g(double x1, x2, x3, x4, x5){ 
  double x; x = 0.0;  (*)
  if (x5 >= 0.5) x = x+ x1*2.0-1.0; (**)
  x = x+ x2*2.0-1.0;
  x = x+ x3*2.0-1.0;
  x = x+ x4*2.0-1.0;  (***)
  return x;} 
\end{Verbatim} 
\end{minipage}\vspace{0.2cm}\\
These programs are deterministic versions of two programs similar to those analysed by Monniaux's experimental analysis\footnote{Personal communication with D.~Monniaux.}~\cite{Monniaux2000}; however, instead of having random generators in the program, here, the random generated inputs are given as program input.
We analyse the programs using a black-box interval analysis corresponding to that lifted by Monniaux~\cite{Monniaux2000} to handle probabilistic programs; furthermore,
to ensure a fair comparison we use an input-partition that corresponds to the abstraction used by Monniaux. 

For both programs we infer probability bounds for the output events $\{[\texttt{-}4,\texttt{-}3],[\texttt{-}3,\texttt{-}2],\ldots,[3,4]\}$, given an input probability measure where
%\begin{table}
the input arguments are independent and each is uniformly distributed between 0 and 1.
Since the interval analysis yields partially correct results, stating nothing about non-termination, we safely assume to include non-termination as part of its output.

\paragraph{Program \texttt{f}.}
The program \texttt{f} return the sum four input variables. Because the interval analysis $\img^{\sharp}_{|\texttt{f}|}$ is a forward analysis, we use the formulas of Theorem~\ref{chp7:thm:simpleBounds} and let the input partition $T$ be the cartesian product $T= I_{1/10}\times I_{1/10}\times I_{1/10}\times I_{1/10}$ where $I_{1/10}=\{[0,\frac{1}{10}],\ldots, [\frac{9}{10},1]\}$, \eg, see following example. 
\begin{example}\label{ex:f:img-examples}
The interval analysis provides an image-overapproximating function $\img^{\sharp}_{|\texttt{f}}|$ where the following are three examples.
$$
\renewcommand*{\arraystretch}{1.2}
\begin{array}{@{\img^{\sharp}_{|\texttt{f}|}(}c@{,}c@{,}c@{,}c@{)\;=\;}c}
% \texttt{x1} & \texttt{x2} & \texttt{x3} & \texttt{x4} & output\\
{[0,\frac{1}{10}]} & {[0,\frac{1}{10}]} & {[0,\frac{1}{10}]} & {[0,\frac{1}{10}]} & \{{[-4,-\frac{32}{10}]}, \bot\}\\
{[0,\frac{1}{10}]} & {[0,\frac{1}{10}]} & {[0,\frac{1}{10}]} & {[\frac{1}{10},\frac{2}{10}]} & \{{[-\frac{38}{10},-\frac{30}{10}]}, \bot\}\\ 
{[0,\frac{1}{10}]} & {[0,\frac{1}{10}]} & {[0,\frac{1}{10}]} & {[\frac{2}{10},\frac{3}{10}]} & \{{[-\frac{36}{10},-\frac{28}{10}]}, \bot\}
%$[0,\frac{1}{3}]$ & $[0,\frac{1}{3}]$ & $[0,\frac{1}{3}]$ & $[\frac{1}{3},\frac{2}{3}]$ & $[0,\frac{1}{3}]$ & $[-2\frac{1}{3},-\frac{1}{3}]$ \\
%$[0,\frac{1}{3}]$ & $[0,\frac{1}{3}]$ & $[0,\frac{1}{3}]$ & $[\frac{1}{3},\frac{2}{3}]$ & $[\frac{1}{3},\frac{2}{3}]$ & $[-3\frac{1}{3},-\frac{1}{3}]$ \\  
\end{array}$$
\end{example}
\noindent There are $10'000$ partition elements, and due to the uniform probability measure each of the partition elements has probability $1/10'000$. We use the formulas from Theorem~\ref{chp7:thm:simpleBounds} to derive upper and lower probability bounds of the output events, \eg, see following example. The inferred probability bounds are shown as black solid lines in Figure~\ref{chp7:fig:monniaux:f:res} and the similar results obtained using Monniaux's experimental analysis~\cite{Monniaux2000} are depicted by the blue dotted lines. 
\begin{example}\label{ex:f:up+low(-4-3)}
We use the formulas from Theorem~\ref{chp7:thm:simpleBounds} to calculate the upper and lower probability bounds of the output event $[-4,-3]$.
There are 70 input partition elements whose over-approximated images, \ie, $\img^{\sharp}_{|\texttt{f}|}$, overlap with the output event $[-4,-3]$; we will refrain from specifying them.
Since all the over-approximated images contain the element $\bot$ then there are 0 input partition elements whose over-approximated images are subsets of $[-4,-3]$.
 The input probability of each element is $1/10'000$ which simplifies the calculations.
$$\begin{array}{lrl}
\mu\supsharp_{|\texttt{f}|}([-4,-3]) = & \sum_{t\in T, \img_{|\texttt{f}|}^{\sharp}(t) \cap [-4,-3] \neq \emptyset}\mu(t) 
 & 
 = 70\cdot \frac{1}{10000} = \frac{7}{1000} \approx 0.007\\
\mu\supflat_{|\texttt{f}|}([-4,-3]) = & \sum_{t\in T, \img''^{\sharp}(t) \subseteq [-4,-3]}\mu(t) & = 
\sum_{t\in \emptyset} \mu(t)  = 0
\end{array}
$$
In comparison the correct probability of the output event $[-4,-3]$ is $1/384\approx 0.00260417\!$ (verified in Mathematica).
\end{example}
\begin{figure}[hb]%\vspace{-0.4cm}
%\begin{subfigure}{0.44\textwidth} \small
%%double orgf() {
%%  double x=0.0;
%%  for (int i=0; i<4; i++) 
%%    x += drand48()*2.0-1.0;
%%  return x; } 
%%  
%\begin{Verbatim}[frame=none,baselinestretch=0.1]
%  double f(double x1, x2, x3, x4){ 
%    double x; x = 0.0;
%    x = x+ x1*2.0-1.0;
%    x = x+ x2*2.0-1.0;
%    x = x+ x3*2.0-1.0;
%    x = x+ x4*2.0-1.0;
%    return x; }
%\end{Verbatim}
%\vspace{-0.6cm}\caption{}\label{chp7:fig:monniaux:f}
%\end{subfigure}
%\begin{subfigure}{0.52\textwidth}\small
%%double orgg() {
%%  double x=0.0;
%%  if(drand48()>0.5)
%%    x += drand48()*2.0-1.0;
%%  for (int i=0; i<4; i++) 
%%    x += drand48()*2.0-1.0;
%%  return x; } 
%%  
%\begin{Verbatim}[frame=none,baselinestretch=0.1]
%double g(double x1, x2, x3, x4, x5){ 
%  double x; x = 0.0;
%  if (x5 >= 0.5) x = x+ x1*2.0-1.0;
%  x = x+ x2*2.0-1.0;
%  x = x+ x3*2.0-1.0;
%  x = x+ x4*2.0-1.0;  
%  return x;} 
%\end{Verbatim} 
%\vspace{-0.6cm}\caption{}\label{chp7:fig:monniaux:g}
%\end{subfigure}
\begin{subfigure}{0.5\textwidth}
\begin{tikzpicture}%[scale=1.1]
      \begin{axis}[%
        height=5cm,
        align =center, title style={yshift=-13pt,font=\footnotesize},
        title={Probabilities and inferred probability bounds \\ (program \texttt{f})},
        xtick={1,2,3,4,5,6,7,8},
        scaled ticks=false,
        ytick={0, 25, 50},
        yticklabels={$0$,$1/4$, $1/2$},%$33\%$,$67\%$,$100\%$},
        ymin=-0.1,
        ymax=60,
        xmin=0.5,
        xmax=9,
        width=1.1\textwidth,
        axis y line*=left,
        axis x line*=bottom,
        line width=0.2mm,
        x tick label style={rotate=-50,anchor=west,font=\footnotesize},
                xticklabels={
 %               ${(-\infty,-4)}$, 
                ${[-4,-3]}$,
                ${[-3,-2]}$, 
                ${[-2,-1]}$,        
                ${[-1,0]}$, 
                ${[0,1]}$, 
                ${[1,2]}$, 
                ${[2,3]}$, 
                ${[3,4]}$,   
 %               ${(4,\infty)}$, 
                \empty}                             
      ]

% "background" grid lines      
%\addplot[very thin,gray,opacity=0.2] coordinates {(-1,1) (16,1)};      
\addplot[very thin,gray,opacity=0.2] coordinates {(-1,50) (10,50)};      
\addplot[very thin,gray,opacity=0.2] coordinates {(-1,25) (10,25)};
\addplot[very thin,gray,opacity=0.2] coordinates {(-1,0) (10,0)};      

%      \addplot[ mark options = solid,mark = -, dashed, mycolor!60!white,shift={(6,0)}] coordinates {(0,0) (0,0)};
      \addplot[ mark options = solid,mark = -, densely dotted, mycolor!60!white,shift={(6,0)}] coordinates {(1,0) (1,0.7)};
      \addplot[ mark options = solid,mark = -, densely dotted, mycolor!60!white,shift={(6,0)}] coordinates {(2,0) (2,7.1)};
      \addplot[ mark options = solid,mark = -, densely dotted, mycolor!60!white,shift={(6,0)}] coordinates {(3,0) (3,25.1)};
      \addplot[ mark options = solid,mark = -, densely dotted, mycolor!60!white,shift={(6,0)}] coordinates {(4,0) (4,46.5)};
      \addplot[ mark options = solid,mark = -, densely dotted, mycolor!60!white,shift={(6,0)}] coordinates {(5,0) (5,46.5)};
      \addplot[ mark options = solid,mark = -, densely dotted, mycolor!60!white,shift={(6,0)}] coordinates {(6,0) (6,25.1)};
      \addplot[ mark options = solid,mark = -, densely dotted, mycolor!60!white,shift={(6,0)}] coordinates {(7,0) (7,7.1)};
      \addplot[ mark options = solid,mark = -, densely dotted, mycolor!60!white,shift={(6,0)}] coordinates {(8,0) (8,0.7)};
%      \addplot[ mark options = solid,mark = -, dashed, mycolor!60!white,shift={(6,0)}] coordinates {(9,0) (9,0)};
      %     
% \addplot[solid,mark = -, mark options = solid, black,shift={(0,0)}] coordinates  {(0,0) (0,0)};
    \addplot[solid, mark = -, mark options = solid, black,shift={(0,0)}] coordinates  {(1,0) (1,0.7)};
    \addplot[solid, mark = -, mark options = solid, black,shift={(0,0)}] coordinates  {(2,0) (2,7.1)};
    \addplot[solid, mark = -, mark options = solid, black,shift={(0,0)}] coordinates  {(3,0) (3,25.1)};
    \addplot[solid, mark = -, mark options = solid, black,shift={(0,0)}] coordinates  {(4,0) (4,46.5)};
    \addplot[solid, mark = -, mark options = solid, black,shift={(0,0)}] coordinates  {(5,0) (5,46.5)};
    \addplot[solid, mark = -, mark options = solid, black,shift={(0,0)}] coordinates  {(6,0) (6,25.1)};
    \addplot[solid, mark = -, mark options = solid, black,shift={(0,0)}] coordinates  {(7,0) (7,7.1)};
    \addplot[solid, mark = -, mark options = solid, black,shift={(0,0)}] coordinates  {(8,0) (8,0.7)};
%    \addplot[solid, mark = -, mark options = solid, black,shift={(0,0)}] coordinates  {(9,0) (9,0)};
%     
 %\addplot[densely dotted,mark = -, mark options = solid, black,shift={(3,0)}] coordinates  {(0,0) (0,0)};
    \addplot[dashed,mark = -, mark options = solid, black,shift={(3,0)}] coordinates  {(1,0.1) (1,0.7)};
    \addplot[dashed,mark = -, mark options = solid, black,shift={(3,0)}] coordinates  {(2,1.4) (2,7.1)};
    \addplot[dashed,mark = -, mark options = solid, black,shift={(3,0)}] coordinates  {(3,6.3) (3,25.1)};
    \addplot[dashed,mark = -, mark options = solid, black,shift={(3,0)}] coordinates  {(4,12.3) (4,46.5)};
    \addplot[dashed,mark = -, mark options = solid, black,shift={(3,0)}] coordinates  {(5,12.3) (5,46.5)};
    \addplot[dashed,mark = -, mark options = solid, black,shift={(3,0)}] coordinates  {(6,6.3) (6,25.1)};
    \addplot[dashed,mark = -, mark options = solid, black,shift={(3,0)}] coordinates  {(7,1.4) (7,7.1)};
    \addplot[dashed,mark = -, mark options = solid, black,shift={(3,0)}] coordinates  {(8,0.1) (8,0.7)};
%    \addplot[densely dotted,mark = -, mark options = solid, black,shift={(3,0)}] coordinates  {(9,0) (9,0)};
%
 %   \addplot[mark=x, mycolor3,shift={(3,0)}] coordinates {(0, 0)   (0, 0)};
    \addplot[mark=x, mycolor3,shift={(3,0)}] coordinates {(1, 0.1) (1, 0.1)};
    \addplot[mark=x, mycolor3,shift={(3,0)}] coordinates {(2, 03)  (2, 03)};
    \addplot[mark=x, mycolor3,shift={(3,0)}] coordinates {(3, 15) (3,15)};
    \addplot[mark=x, mycolor3,shift={(3,0)}] coordinates {(4, 32) (4, 32)};
    \addplot[mark=x, mycolor3,shift={(3,0)}] coordinates {(5, 32) (5, 32)};
    \addplot[mark=x, mycolor3,shift={(3,0)}] coordinates {(6, 15) (6, 15)};
    \addplot[mark=x, mycolor3,shift={(3,0)}] coordinates {(7, 03) (7, 03)};
    \addplot[mark=x, mycolor3,shift={(3,0)}] coordinates {(8,0.1) (8,0.1)};
%    \addplot[mark=x, mycolor3,shift={(3,0)}] coordinates {(9,0) (9,0)};
       \end{axis}
    \end{tikzpicture}%
\vspace{-0.5cm}\caption{}\label{chp7:fig:monniaux:f:res}
\end{subfigure}
\begin{subfigure}{0.5\textwidth}
    \noindent\begin{tikzpicture}%[scale=1.1]
      \begin{axis}[%
        height=5 cm,
        align =center, title style={yshift=-13pt,  font=\footnotesize},
        title={Probabilities and inferred probability bounds \\ (program \texttt{g})},
        xtick={1,2,3,4,5,6,7,8},
        scaled ticks=false,
        ytick={0, 33, 66, 100},
        yticklabels={$0$,$1/3$,$2/3$,$1$},
        ymin=-0.1,
        ymax=120,
        xmin=0.5,
        xmax=8.5,
        width=\textwidth,
        axis y line*=left,
        axis x line*=bottom,
        line width=0.2mm,
        x tick label style={rotate=-50,anchor=west,font=\footnotesize},
                xticklabels={
               % \empty,%${(-\infty,-4)}$, 
                ${[-4,-3]}$,
                ${[-3,-2]}$, 
                ${[-2,-1]}$,        
                ${[-1,0]}$, 
                ${[0,1]}$, 
                ${[1,2]}$, 
                ${[2,3]}$, 
                ${[3,4]}$%,   
                %${(4,\infty)}$, 
                %\empty
                }                             
      ]

% "background" grid lines      
\addplot[very thin,gray,opacity=0.2] coordinates {(-1,100) (10,100)};      
\addplot[very thin,gray,opacity=0.2] coordinates {(-1,66.666) (10,66.666)};      
\addplot[very thin,gray,opacity=0.2] coordinates {(-1,33.333) (10,33.333)};
\addplot[very thin,gray,opacity=0.2] coordinates {(-1,0) (10,0)};      

      \addplot[ mark options = solid,mark = -, densely dotted, mycolor!60!white,shift={(1.5,0)}] coordinates {(0,0) (0,0)};
      \addplot[ mark options = solid,mark = -, densely dotted, mycolor!60!white,shift={(1.5,0)}] coordinates {(1,0) (1,4)};
      \addplot[ mark options = solid,mark = -, densely dotted, mycolor!60!white,shift={(1.5,0)}] coordinates {(2,0) (2,22)};
      \addplot[ mark options = solid,mark = -, densely dotted, mycolor!60!white,shift={(1.5,0)}] coordinates {(3,0) (3,66)};
      \addplot[ mark options = solid,mark = -, densely dotted, mycolor!60!white,shift={(1.5,0)}] coordinates {(4,0) (4,108)};
      \addplot[ mark options = solid,mark = -, densely dotted, mycolor!60!white,shift={(1.5,0)}] coordinates {(5,0) (5,108)};
      \addplot[ mark options = solid,mark = -, densely dotted, mycolor!60!white,shift={(1.5,0)}] coordinates {(6,0) (6,66)};
      \addplot[ mark options = solid,mark = -, densely dotted, mycolor!60!white,shift={(1.5,0)}] coordinates {(7,0) (7,22)};
      \addplot[ mark options = solid,mark = -, densely dotted, mycolor!60!white,shift={(1.5,0)}] coordinates {(8,0) (8,4)};
      \addplot[ mark options = solid,mark = -, densely dotted, mycolor!60!white,shift={(1.5,0)}] coordinates {(9,0) (9,0)};
     
 \addplot[solid,mark = -, mark options = solid, black,shift={(-1.5,0)}] coordinates  {(0,0) (0,0)};
    \addplot[solid,mark = -, mark options = solid, black,shift={(-1.5,0)}] coordinates  {(1,0) (1,4)};
    \addplot[solid,mark = -, mark options = solid, black,shift={(-1.5,0)}] coordinates  {(2,0) (2,19)};
    \addplot[solid,mark = -, mark options = solid, black,shift={(-1.5,0)}] coordinates  {(3,0) (3,53)};
    \addplot[solid,mark = -, mark options = solid, black,shift={(-1.5,0)}] coordinates  {(4,0) (4,83)};
    \addplot[solid,mark = -, mark options = solid, black,shift={(-1.5,0)}] coordinates  {(5,0) (5,83)};
    \addplot[solid,mark = -, mark options = solid, black,shift={(-1.5,0)}] coordinates  {(6,0) (6,53)};
    \addplot[solid,mark = -, mark options = solid, black,shift={(-1.5,0)}] coordinates  {(7,0) (7,19)};
    \addplot[solid,mark = -, mark options = solid, black,shift={(-1.5,0)}] coordinates  {(8,0) (8,4)};
    \addplot[solid,mark = -, mark options = solid, black,shift={(-1.5,0)}] coordinates  {(9,0) (9,0)};
    \addplot[mark=x, mycolor3,shift={(-1.5,0)}] coordinates {(0, 0)   (0, 0)};
    \addplot[mark=x, mycolor3,shift={(-1.5,0)}] coordinates {(1, 0.1) (1, 0.1)};
    \addplot[mark=x, mycolor3,shift={(-1.5,0)}] coordinates {(2, 03)  (2, 03)};
    \addplot[mark=x, mycolor3,shift={(-1.5,0)}] coordinates {(3, 15) (3,15)};
    \addplot[mark=x, mycolor3,shift={(-1.5,0)}] coordinates {(4, 32) (4, 32)};
    \addplot[mark=x, mycolor3,shift={(-1.5,0)}] coordinates {(5, 32) (5, 32)};
    \addplot[mark=x, mycolor3,shift={(-1.5,0)}] coordinates {(6, 15) (6, 15)};
    \addplot[mark=x, mycolor3,shift={(-1.5,0)}] coordinates {(7, 03) (7, 03)};
    \addplot[mark=x, mycolor3,shift={(-1.5,0)}] coordinates {(8,0.1) (8,0.1)};
    \addplot[mark=x, mycolor3,shift={(-1.5,0)}] coordinates {(9,0) (9,0)};
       \end{axis}
    \end{tikzpicture}
    \vspace{-0.4cm}\caption{}\label{chp7:fig:monniaux:g:res}
\end{subfigure}   \vspace{-0.4cm} 
    \caption{The black solid lines are the inferred probability bounds obtained using an interval analysis. In (\subref{chp7:fig:monniaux:f:res}), the dashed black lines show the improved inferred probability bounds based on combining the interval-analysis results with termination-analysis results. The blue dotted lines indicates the upper probability bounds for the output events derived using Monniaux's example analysis~\cite{Monniaux2000}.} \vspace{-0.3cm} 
    \label{chp7:fig:monniaux:1}
\end{figure}
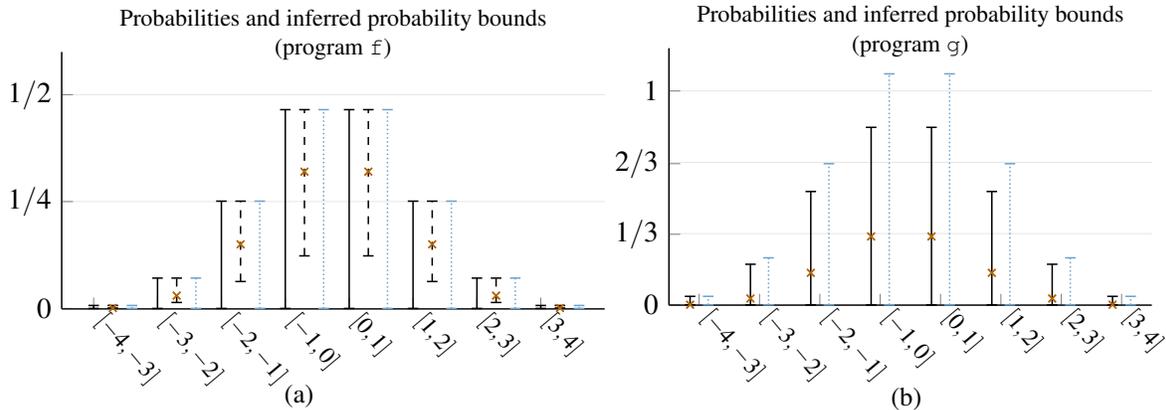
\noindent \emph{Combined analyses.} As we saw in the sign-analysis example (Section~\ref{sec:sign}), combining the analysis with the results of a termination analysis may improve the results. This is also the case in this example as shown by the black dotted lines in Figure~\ref{chp7:fig:monniaux:f:res}, e.g., see following example.
\begin{example}
As for $\img^{\sharp}_{|\texttt{f}|}$ the same 70 partition elements overlaps with the output event $[-4,-3]$. 
However, there are 5 input partition elements for which $\img'^{\sharp}_{|\texttt{f}|}$ are subsets of $[-4,-3]$, namely 
$T' = \{(i,i,i,i), (j,i,i,i), (i,j,i,i), 
(i,i,j,i), (i,i,i,j)\}$ where $i=[0,\frac{1}{10}]$ and  $j = [\frac{1}{10},\frac{2}{10}]\}$, as the first two calculations   in Example~\ref{ex:f:img-examples} indicate.
 The input probability of each element is $1/10'000$ which simplifies the calculations.
$$\begin{array}{lrl}
\mu\supsharp_{|\texttt{f}|}([-4,-3]) = & \sum_{t\in T, \img_{|\texttt{f}|}^{\sharp}(t) \cap [-4,-3] \neq \emptyset}\mu(t) 
 & 
 = 70\cdot \frac{1}{10000} = \frac{7}{1000} \approx 0.007\\
\mu\supflat_{|\texttt{f}|}([-4,-3]) = & \sum_{t\in T, \img''^{\sharp}(t) \subseteq [-4,-3]}\mu(t) & = 
\sum_{t\in T'}\frac{1}{10000} = 5\cdot\frac{1}{10000} =\frac{5}{10000} \approx 0.0005
\end{array}
$$
In comparison the correct probability of the output event $[-4,-3]$ is $1/384\approx 0.00260417$ (verified in Mathematica).
\end{example}
The non-trivial lower probability bounds is a novel development, and Monniaux's  lifting framework is created with the purpose of deriving only upper probability bounds. 

\paragraph{Program \texttt{g}.}
The program \texttt{g} (\pageref{sec:casestudyinterval}) returns the sum three or four input variables, depending on the value of the fifth input variable. Due to the branching point in \texttt{g}, \ie, the if-expression, we see a difference in the upper probability bounds derived using the presented approach and using Monniaux's experimental analysis.

Again, the interval analysis $\img^{\sharp}_{|\texttt{g}|}$ is a forward analysis and we have used the formulas of Theorem~\ref{chp7:thm:simpleBounds}. The difference in the bounds stems from undecidability occurring in the branching point and to expose this difference we have chosen a partition causing such undecidability.
The input partition is the cartesian product $I_{1/3}\times I_{1/3}\times I_{1/3}\times I_{1/3}\times I_{1/3}$ where $I_{1/3}=\{[0,\frac{1}{3}],[\frac{1}{3},\frac{2}{3}],[\frac{2}{3},1]\}$. There are $243$ partition elements, and due to the uniform probability measure each of the partition elements has probability $1/243$. 
When comparing the results of the methods presented in this paper, shown by the solid black lines in Figure~\ref{chp7:fig:monniaux:g:res}, and the results obtained using Monniaux's experimental analysis~\cite{Monniaux2000}, shown by the blue dotted lines in Figure~\ref{chp7:fig:monniaux:g:res}, the methods presented here produce better results.\\

\noindent \emph{The difference.} In the following we demonstrate via an example what causes the difference between our results and those derived using Monniaux's experimental analysis~\cite{Monniaux2000}. Instead of studying the above partition with 243 elements, we study a smaller input partition with 6 elements, namely 
$T = \{{[0,1]}\} \times \{{[0,1]}\}\times \{{[0,1]}\}\times \{{[0,\frac{1}{2}]},{[\frac{1}{2},1]}\} \times \{{[0,\frac{1}{3}]},{[\frac{1}{3},\frac{2}{3}]},{[\frac{2}{3},1]}\}$
%
%
% $T= ({[\frac{2}{3},1]},{[\frac{2}{3},1]},{[\frac{2}{3},1]},{[\frac{2}{3},1]})
%({[0,1]}, {[0,1]}, {[0,1]}, {[0,1]}, {[0,\frac{1}{3}]}), 
%({[0,1]}, {[0,1]}, {[0,1]}, {[0,1]}, {[\frac{1}{3},\frac{2}{3}]}), 
%({[0,1]}, {[0,1]}, {[0,1]}, {[0,1]}, {[\frac{2}{3},1]}) 
%\}$,
 where each partition element has input probability $1/6$ and we will derive the upper probability bound for the output event $[2\frac{1}{2},3\frac{1}{2}]$.
\begin{example}\label{ex:g:img-examples}
The interval analysis provides an image-overapproximating function $\img^{\sharp}_{|\texttt{g}|}$ over partition $T$:
$$
\renewcommand*{\arraystretch}{1.2}
\begin{array}{@{\img^{\sharp}_{|\texttt{g}|}(}c@{,}c@{,}c@{,}c@{,}c@{)\;=\;}c}
% \texttt{x1} & \texttt{x2} & \texttt{x3} & \texttt{x4} & output\\
{[0,1]} & {[0,1]} &{[0,1]} & {[0,\frac{1}{2}]} & {[0,\frac{1}{3}]} & \{{[-3,2]}, \bot\}\\
 {[0,1]} & {[0,1]} & {[0,1]} &{[0,\frac{1}{2}]} & {[\frac{1}{3},\frac{2}{3}]} & \{{[-4,3]}, \bot\}\\
 {[0,1]} & {[0,1]} & {[0,1]} & {[0,\frac{1}{2}]} &{[\frac{2}{3},1]} & \{{[-4,3]}, \bot\}\\
{[0,1]} & {[0,1]} & {[0,1]} & {[\frac{1}{2},1]} & {[0,\frac{1}{3}]} & \{{[-2,3]}, \bot\}\\
 {[0,1]} & {[0,1]} & {[0,1]} & {[\frac{1}{2},1]} & {[\frac{1}{3},\frac{2}{3}]} & \{{[-3,4]}, \bot\}\\
 {[0,1]} & {[0,1]} & {[0,1]} & {[\frac{1}{2},1]} & {[\frac{1}{2},1]} & \{{[-3,4]}, \bot\}
\end{array}$$
Similar to the previous examples, we use Theorem~\ref{chp7:thm:simpleBounds} to infer an upper probability bound for output event $[2\frac{1}{2},3\frac{1}{2}]$.
\begin{align*}
\mu\supsharp_{|\texttt{f}|}([2\frac{1}{2},3\frac{1}{2}]) = & \sum_{t\in T, \img_{|\texttt{g}|}^{\sharp}(t) \cap [2\frac{1}{2},3\frac{1}{2}] \neq \emptyset}\mu(t) = \sum_{t\in T'}\mu(t) = 
5/6.
\end{align*}
where $T' = T \setminus \{({[0,1]},{[0,1]},{[0,1]},{[0,\frac{1}{2}]},{[0,\frac{1}{3}])}\}$.
\end{example}
Monniaux's experimental abstracts the input measure to a set containing pairs of interval environments  (similar to the interval analysis) and a weight, \ie, the input probability of that interval environment, \eg, $\langle E, w \rangle$,. Afterwards, it propagates the pairs through the program according to the specifications of the interval analysis; at the branching points the lifted interval analysis duplicates the pair, \eg, $\langle E, w \rangle$ and $\langle E, w \rangle$,  and adjust the environments according to the condition and its negation to avoid infeasible environments, \eg, $\langle E_1, w \rangle$ and $\langle E_2, w \rangle$, \eg, see following example. When the set is propagated all the way through the program, the probability of an output event is the sum of the weights of those environments yielding output which overlaps with the output event. 

In the following we propagate the pairs through the program and calculates the upper probability bound; we provide the set of pairs at the program points at \texttt{*}, \texttt{**}, and \texttt{***}. 
\begin{example}
At program point \texttt{*} in \texttt{g} (p.~\pageref{sec:casestudyinterval}) the environments are still recognizable from the input partition elements with the addition of variable \texttt{\emph{x}}. The interval of \texttt{\emph{x5}} in the second and fifth pair may lead to the condition being true or false (undecidable); this causes a split of each the those pairs into two whose \texttt{\emph{x5}} intervals are updated according to the branch condition. Afterwards, \texttt{\emph{x}} is updated as expected; according to the branch (or skip) before reaching program point \texttt{**}), and again upon reaching program point \texttt{***}.
In the following, we have only included the essential parts of the environments essential, \eg, at every program point we have left out $\emph{\texttt{x1}} \mapsto {[0,1]}, \emph{\texttt{x2}} \mapsto {[0,1]}, \emph{\texttt{x3}} \mapsto {[0,1]}$. 
%$$
%\begin{array}{l@{}l}
%\multicolumn{2}{l}{\text{At program point \texttt{*}:}}\\
%\{\langle \ldots, \texttt{x5} \mapsto {[0,\frac{1}{3}], \texttt{x} \mapsto {[0,0]}}; \frac{1}{3}\rangle, &
%\langle \ldots, \texttt{x5} \mapsto {[\frac{1}{3},\frac{2}{3}], \texttt{x} \mapsto {[0,0]}}; \frac{1}{3}\rangle,  
%\langle \ldots, \texttt{x5} \mapsto {[\frac{2}{3},1], \texttt{x} \mapsto {[0,0]}}; \frac{1}{3}\rangle\}
%\\
%\multicolumn{2}{l}{\text{At program point \texttt{**}, :}}\\
%\{\langle \ldots, \texttt{x5} \mapsto {[0,\frac{1}{3}], \texttt{x} \mapsto {[0,0]}}; \frac{1}{3}\rangle, &\langle \ldots,\texttt{x5} \mapsto {[\frac{1}{3},\frac{1}{2}], \texttt{x} \mapsto {[0,0]}}; \frac{1}{3}\rangle,  \\
%&\langle \ldots, \texttt{x5} \mapsto {[\frac{1}{2},\frac{2}{3}], \texttt{x} \mapsto {[-1,1]}}; \frac{1}{3}\rangle, 
%\langle \ldots, \texttt{x5} \mapsto {[\frac{2}{3},1], \texttt{x} \mapsto {[-1,1]}}; \frac{1}{3}\rangle\}\\
%\multicolumn{2}{l}{\text{at program point \texttt{***}:}}\\
%\{\langle \ldots, \texttt{x5} \mapsto {[0,\frac{1}{3}], \texttt{x} \mapsto {[-3,3]}}; \frac{1}{3}\rangle,  
%&\langle \ldots, \texttt{x5} \mapsto {[\frac{1}{3},\frac{1}{2}], \texttt{x} \mapsto {[-3,3]}};\frac{1}{3}\rangle,  \\
%&\langle \ldots, \texttt{x5} \mapsto {[\frac{1}{2},\frac{2}{3}], \texttt{x} \mapsto {[-4,4]}}; \frac{1}{3}\rangle,  \langle \ldots, \texttt{x5} \mapsto {[\frac{2}{3},1], \texttt{x} \mapsto {[-4,4]}}; \frac{1}{3}\rangle\}
%\end{array}$$
%
$$\renewcommand*{\arraystretch}{1.2}
\begin{array}{@{}lll}
\text{Program point \texttt{*}:} & Program point \text{\texttt{**}:} & \text{Program point \texttt{***}:} \\
\begin{array}{@{}l@{}}
\{\!\langle \mydots, \texttt{\emph{x4}} {\mapsto} [0,\frac{1}{2}], \texttt{\emph{x5}} {\mapsto} [0,\frac{1}{3}], \texttt{\emph{x}} {\mapsto} [0,0] ; \frac{1}{6}\rangle, \\ \phantom{\{\!}
\langle \mydots, \texttt{\emph{x4}} {\mapsto} [0,\frac{1}{2}],  \texttt{\emph{x5}} {\mapsto} [\frac{1}{3},\frac{2}{3}], \texttt{\emph{x}} {\mapsto} [0,0] ; \frac{1}{6}\rangle, \\ \\ \phantom{\{\!}
\langle \mydots, \texttt{\emph{x4}} {\mapsto} [0,\frac{1}{2}],  \texttt{\emph{x5}} {\mapsto} [\frac{2}{3},1], \texttt{\emph{x}} {\mapsto} [0,0] ; \frac{1}{6}\rangle, \\ \phantom{\{\!}
\langle \mydots, \texttt{\emph{x4}} {\mapsto} [\frac{1}{2},1],   \texttt{\emph{x5}} {\mapsto} [0,\frac{1}{3}], \texttt{\emph{x}} {\mapsto} [0,0] ; \frac{1}{6}\rangle, \\  \phantom{\{\!}
\langle \mydots, \texttt{\emph{x4}} {\mapsto} [\frac{1}{2},1],  \texttt{\emph{x5}} {\mapsto} [\frac{1}{3},\frac{2}{3}], \texttt{\emph{x}} {\mapsto} [0,0] ; \frac{1}{6}\rangle, \\ \\ \phantom{\{\!}
\langle \mydots, \texttt{\emph{x4}} {\mapsto} [\frac{1}{2},1],   \texttt{\emph{x5}} {\mapsto} [\frac{2}{3},1], \texttt{\emph{x}} {\mapsto} [0,0] ; \frac{1}{6}\rangle \!\}\\
\end{array}&
\begin{array}{@{}l@{}}
\{\!\langle \mydots, \texttt{\emph{x5}} {\mapsto} [0,\frac{1}{3}], \texttt{\emph{x}} {\,\mapsto} [0,0] ; \phantom{-} \frac{1}{6}\rangle, \\ \phantom{\{\!}
\langle \mydots, \texttt{\emph{x5}} {\mapsto} [\frac{1}{3},\frac{1}{2}], \texttt{\emph{x}} {\,\mapsto} [0,0] ; \phantom{-} \frac{1}{6}\rangle, \\ \phantom{\{\!}
\langle \mydots, \texttt{\emph{x5}} {\mapsto} [\frac{1}{2},\frac{2}{3}], \texttt{\emph{x}} {\,\mapsto} [-1,1] ; \frac{1}{6}\rangle, \\ \phantom{\{\!}
\langle \mydots, \texttt{\emph{x5}} {\mapsto} [\frac{2}{3},1], \texttt{\emph{x}} {\,\mapsto} [-1,1] ; \frac{1}{6}\rangle, \\ \phantom{\{\!}
\langle \mydots, \texttt{\emph{x5}} {\mapsto} [0,\frac{1}{3}], \texttt{\emph{x}} {\,\mapsto} [0,0] ; \phantom{-} \frac{1}{6}\rangle, \\ \phantom{\{\!}
\langle \mydots, \texttt{\emph{x5}} {\mapsto} [\frac{1}{3},\frac{1}{2}], \texttt{\emph{x}} {\,\mapsto} [0,0] ; \phantom{-} \frac{1}{6}\rangle, \\ \phantom{\{\!}
\langle \mydots, \texttt{\emph{x5}} {\mapsto} [\frac{1}{2},\frac{2}{3}], \texttt{\emph{x}} {\,\mapsto} [-1,1] ; \frac{1}{6}\rangle, \\ \phantom{\{\!}
\langle \mydots, \texttt{\emph{x5}} {\mapsto} [\frac{2}{3},1], \texttt{\emph{x}} {\,\mapsto} [-1,1] ; \frac{1}{6}\rangle\!\}\\
\end{array}&
\begin{array}{@{}l@{}}
\{\!\langle \mydots, \texttt{\emph{x}} {\,\mapsto} [-3,2] ; \frac{1}{6}\rangle, \\ \phantom{\{\!}
\langle \mydots, \texttt{\emph{x}} {\,\mapsto} [-3,2] ; \frac{1}{6}\rangle, \\ \phantom{\{\!}
\langle \mydots, \texttt{\emph{x}} {\,\mapsto} [-4,3] ; \frac{1}{6}\rangle, \\ \phantom{\{\!}
\langle \mydots,  \texttt{\emph{x}} {\,\mapsto} [-4,3]; \frac{1}{6} \rangle, \\ \phantom{\{\!}
\langle \mydots, \texttt{\emph{x}} {\,\mapsto} [-2,3] ; \frac{1}{6}\rangle, \\ \phantom{\{\!}
\langle \mydots,  \texttt{\emph{x}} {\,\mapsto} [-2,3]; \frac{1}{6} \rangle, \\ \phantom{\{\!}
\langle \mydots,  \texttt{\emph{x}} {\,\mapsto} [-3,4]; \frac{1}{6} \rangle, \\ \phantom{\{\!}
\langle \mydots,  \texttt{\emph{x}} {\,\mapsto} [-3,4]; \frac{1}{6} \rangle\!\}\\
\end{array}
\end{array}
$$
The upper probability bound of output event $[2\frac{1}{2},3\frac{1}{2}]$ is the sum of the weights of the pairs for which the \texttt{x}'s value overlaps with $[2\frac{1}{2},3\frac{1}{2}]$ at program point \texttt{***}. The last six pairs are the only ones where \texttt{x} interval overlap with $[2\frac{1}{2},3\frac{1}{2}]$, thus, Monniaux's experimental analysis~\cite{Monniaux2000} derives $6\cdot \frac{1}{6} = 1$ as the upper probability bound for $[2\frac{1}{2},3\frac{1}{2}]$. Thus, the bound ($5/6$) derived by the presented technique is tighter.
\end{example}
\FloatBarrier
\section{Conclusion}\label{sec:conclusion}
%We derived a pair of upper and lower probability bounds for output events, therein using the input probability distribution and a pre-image over-approximating function $\pre^{\sharp}$, \ie{} a function such that $\pre(A) \subseteq \pre^{\sharp}(A)$ whenever  $A$ is an output event.
%We focused on the idea of ``reusing existing analyses'', either forward or backwards, to obtain the necessary pre-image over-approximating function. We presented two techniques, one for each case, and exemplified the forward case by minor examples produced systematically using implemented analysis and calculating the probability bounds by hand.
%
We have presented two simple techniques for reusing existing (non-probabilistic) analyses to derive upper and lower probability bounds of output events; introducing abstraction when the pre-image is non-measurable. We demonstrated forward technique and the initial results are powerful compared to more complex analyses.
%The core of the techniques is a straight-forward formalization of over/under-approximation  of the pre-image operator which is used to derive the upper and lower bounds of output events over the given input probability measure.

%The techniques for forward analysis was demonstrated by an interval analysis and a termination analysis. In addition, we demonstrated that their combined analysis may obtain more accurate bounds. 
%
%
%The paper discuss both practical and theoretical challenges such as (i) how to formally ensure that overapproximating analyses yield results (sets) that are
%measurable in the input distribution; (ii) how does one design abstract domains suitable for the task of probabilistic
%interpretation of results; and (iii), how to ensure that the sigma algebra over the inputs constitute a complete lattice.
%
%The technique for forward analysis was demonstrated by a sign analysis, an interval analysis and a termination analysis. In addition, we demonstrated that their combined analysis may obtain more accurate bounds. 
%
%The presented techniques assumed analyses yielding functions that over-approximated the pre-images or images. Lemma~\ref{chp7:lem:dualunderapproxpre} indicates a similar approach for analyses yielding functions that under-approximate the pre-images or images. This is a part of future work together with generalizing the theorems for nondeterministic programs. 

%\nocite{*}
\FloatBarrier
\bibliographystyle{eptcs}
%\bibliography{generic}
\bibliography{Updated}
\end{document}